\documentclass[12pt]{article}
\usepackage[utf8]{inputenc}
\usepackage{mathtools}
\usepackage{booktabs}
\usepackage[english]{babel} 
\usepackage[protrusion=true,expansion=true]{microtype} 
\usepackage{amsmath,amsfonts,amsthm}
\usepackage{ amssymb }
\usepackage{graphicx}
\usepackage{array}
\usepackage{multirow}
\usepackage{hhline}
\usepackage[titletoc]{appendix}

\usepackage{url}

\usepackage{booktabs} 
\usepackage[numbers]{natbib}
\usepackage{setspace}
\usepackage[margin=1in]{geometry}

\usepackage{tabularx}
\usepackage[font = small,labelfont=bf,textfont=it]{caption} 
\usepackage{subcaption}
\usepackage{footnote}
\usepackage{algorithm}
\usepackage[noend]{algpseudocode}

\newtheorem{theorem}{Theorem}

\newtheorem{lemma}{Lemma}
\newtheorem{definition}{Definition}
\newtheorem{assumption}{Assumption}
\newtheorem{proposition}{Proposition}
\newtheorem{corollary}{Corollary}
\newtheorem{remark}{Remark}

\theoremstyle{definition}

\newcommand{\RR}{\mathbb{R}}

\newcommand{\argmax}{\operatorname*{argmax}}

\begin{document}
\title {Eigenvector-based sparse canonical correlation analysis: Fast computation for estimation of multiple canonical vectors}
\author{Wenjia Wang\\Hong Kong University of Science and Technology\\
Clear Water Bay, Kowloon, Hong Kong \\ \and Yi-Hui Zhou\\Department of Biological Sciences\\
North Carolina State University, Raleigh, NC, U.S.A. } 
\date{}
\maketitle
\vspace{-5mm}

\begin{abstract}
Classical canonical correlation analysis (CCA) requires matrices to be low dimensional, i.e. the number of features cannot exceed the sample size. Recent developments in CCA have mainly focused on the high-dimensional setting, where the number of features in both matrices under analysis greatly exceeds the sample size. These approaches impose penalties in the optimization problems that are needed to be solve iteratively, and estimate multiple canonical vectors sequentially. In this work, we provide an explicit link between sparse multiple regression with sparse canonical correlation analysis, and an efficient algorithm that can estimate multiple canonical pairs simultaneously rather than sequentially. Furthermore, the algorithm naturally allows parallel computing. These properties make the algorithm much efficient. We provide theoretical results on the consistency of canonical pairs. The algorithm and theoretical development are based on solving an eigenvectors problem, which significantly differentiate our method with existing methods. Simulation results support the improved performance of the proposed approach. We apply eigenvector-based CCA to analysis of the GTEx thyroid histology images, analysis of SNPs and RNA-seq gene expression data, and a microbiome study. The real data analysis also shows improved performance compared to traditional sparse CCA.
\end{abstract}

\section{Introduction}\label{secIntro}

Canonical correlation analysis (CCA) is a widely used method to determine the relationship between two sets of variables. In CCA, the objective is to find linear combinations of variables from each set of variables such that the correlation is maximized. The vectors consisting of coefficients from each linear combination are called canonical pairs. Originally proposed by \cite{hotelling1936relations}, CCA has been applied to numerous problems, including those of large scale. In large scale problems, including genomic studies \cite{chen2012structure,wang2015inferring}, medicine \cite{samarov2011local,yazici2010application}, natural language processing \cite{vinokourov2003inferring, haghighi2008learning}, and multimodal signal processing \cite{gao2017discriminative, sargin2007audiovisual}, researchers are often faced with high dimensional data. Projects such as GTEx \cite{aguet2019gtex} also provide rich datasets (and image data) for which CCA might be used to identify important genetic modules relevant to disease. In these works, classical CCA cannot be used to analyze the high dimensional data, where the number of variables exceeds the number of observations. 

To study the relationship between two sets of high dimensional variables, many extensions of classical CCA have been proposed. One popular approach, sparse canonical correlation analysis, imposes sparse structure on the canonical vectors. An incomplete list of sparse CCA methods is \cite{parkhomenko2009sparse,witten2009extensions,waaijenborg2008quantifying,le2009sparse,witten2009penalized,chen2012structure}, and references therein. In the sparse canonical correlation analysis, the canonical pairs are estimated sequentially. Recent works on PCA-CCA \cite{song2016canonical} and decomposition-based CCA \cite{shu2020d} allow to estimate multiple correlation pairs simultaneously, which is more efficient, yet under different assumptions named as ``low-rank plus noise''. 

In this work, we propose \textit{eigenvector-based sparse canonical correlation analysis} (E-CCA). Specifically, we link sparse multiple regression and Lasso regularization with sparse CCA, and solve an eigenvector problem to obtain the canonical pairs. We propose an efficient algorithm to provide $K$ canonical pairs simultaneously for $K\geq 1$. This advantage significantly differentiates our method from other sparse CCA methods, which usually estimate multiple canonical pairs sequentially. The algorithm allows parallel computing, which, together with estimating multiple canonical pairs simultaneously, makes the computation very fast. We also provide theoretical guarantees on the consistency of estimated canonical pairs under the assumptions similar to those in sparse CCA. 

We note that the relationship between multiple regression and CCA has been considered previously. In \cite{glahn1968canonical}, multiple regression was considered as be a special case of CCA, but the high dimensional situation was not considered. \cite{lutz1994relationship} analyzed the relationship between multiple regression and CCA via eigenstructure. \cite{yamamoto2008canonical} applied CCA to multivariate regression. \cite{song2016canonical} assumed that the responses have a linear relationship with some underlying signals. However, we are not aware of any works that apply sparse multiple regression to canonical correlation analysis.

The rest of this paper is arranged as follows. In Section \ref{secbackground}, we introduce classical canonical correlation analysis, sparse canonical correlation analysis, and other canonical correlation analysis methods. We propose an eigenvector-based sparse canonical correlation analysis approach, with attendant theoretical properties in Section \ref{seconelasso}. In Section \ref{secnumeric}, we conduct numeric simulation studies. In Section \ref{seccasestudy}, we apply eigenvector-based sparse canonical correlation analysis and competing methods to three applied problems, including GTEx thyroid imaging/expression data, GTEx liver genotype/expression data, and human gut microbiome data. A discussion with conclusions is provided in Section \ref{secconclu}. Proof of main results are presented in Section \ref{apppfonelasso}.

\section{Preliminaries}\label{secbackground}
In this section, we provide a brief introduction to related canonical correlation analysis methods, including classical canonical correlation analysis, sparse canonical correlation analysis, and other canonical correlation analysis.

\subsection{Classical canonical correlation analysis}\label{secclasscca}
Suppose we are interested in studying the correlation between two sets of random variables $x=(x_1,...,x_p)^\top \in \RR^p$ and $y=(y_1,...,y_d)^\top \in \RR^d$. Given $1\leq K\leq \min\{p,d\}$, the goal of canonical correlation analysis (CCA) is to find $a_1,...,a_K\in \RR^d$ and $b_1,...,b_K\in \RR^p$ such that $(a_1,b_1)$ is the solution to
\begin{align}\label{CCAformpre}
    \max_{a\in \RR^d,b\in \RR^p} &\mbox{ Corr}(a^\top y,b^\top x),
\end{align}
and $(a_k,b_k)$ is the solution to
\begin{align}\label{CCAformpreK}
    \max_{a\in \RR^d,b\in \RR^p} &\mbox{ Corr}(a^\top y,b^\top x),\nonumber\\
   {\rm s.t.} & \mbox{ Corr}(a^\top y,a_l^\top y)=\mbox{ Corr}(b^\top y,b_l^\top y)=0, \forall 1\leq l\leq k-1,
\end{align}
for $k\in\{2,...,K\}$.
Without loss of generality, we assume $x$ and $y$ have mean zero, for otherwise we can shift the mean. Let $\Sigma_{xx}$ and $\Sigma_{yy}$ be the covariance matrix of $x$ and $y$, respectively. Let $\Sigma_{xy}$ be the covariance matrix between $x$ and $y$. The optimization problem \eqref{CCAformpre} is the same as
\begin{align}\label{CCAform}
    \max_{a\in \RR^d,b\in \RR^p} & \frac{a^\top \Sigma_{yx}b}{\sqrt{a^\top \Sigma_{yy}a}\sqrt{b^\top \Sigma_{xx}b}}.
\end{align}
The solution to \eqref{CCAform}, denoted by $a_1$ and $b_1$, are called the first pair of canonical vectors, and the new variables $x_1' = a_1^\top x$ and $y_1'=b_1^\top y$ are called  the first pair of canonical variables \cite{mardia1979multivariate}. Once the $(k-1)$ pairs of canonical vectors $a_1,...,a_{k-1}$ and $b_1,...,b_{k-1}$ are obtained, the $k$-th pair of canonical vectors is the solution to the optimization problem \eqref{CCAformpreK}, which is the same as
\begin{align}\label{CCAformkpair}
    \max_{a\in \RR^d,b\in \RR^p} & \frac{a^\top \Sigma_{yx}b}{\sqrt{a^\top \Sigma_{yy}a}\sqrt{b^\top \Sigma_{xx}b}}\nonumber\\
    {\rm s.t.} & \quad a^\top \Sigma_{yy}a_l = 0, b^\top \Sigma_{xx}b_l = 0, 1\leq l \leq k-1,
\end{align}
for $k\in\{2,...,K\}$.

By basic matrix computation, one can obtain that the solution $a_k$ to the optimization problem \eqref{CCAformkpair} is the $k$-th eigenvector of
\begin{align}\label{traforma}
     \Sigma_{yy}^{-1}\Sigma_{yx}\Sigma_{xx}^{-1}\Sigma_{xy},
\end{align}
and $b_k$ is proportional to 
\begin{align}\label{traformb}
    \Sigma_{xx}^{-1}\Sigma_{xy}a_k.
\end{align}
Note that the solution to \eqref{CCAformkpair} is not unique, because for any constant $C\in\RR$ and $C\neq 0$, if $(a_k,b_k)$ is the solution to \eqref{CCAformkpair}, then so is $(Ca_k,Cb_k)$. Therefore, we restrict the norms of $a_k$ and $b_k$ such that $\|a_k\|_2 = \|b_k\|_2 = 1$, and the first nonzero element of $a_k$ ($b_k$) is positive to make the solution to \eqref{CCAformkpair} unique, where $\|\cdot\|_2$ is the Euclidean norm. This restriction is not essential because one can always scale the canonical vectors $a_k$ and $b_k$ such that $a_k$ and $b_k$ satisfy other constraints, for example, $a_k^\top \Sigma_{yy}a_k=b_k^\top \Sigma_{xx}b_k=1$.

Let $X_i,Y_i$, $i\in\{1,...,n\}$ be observations, where $X_i = (x_{i1},...,x_{ip})^\top $ and $Y_i = (y_{i1},...,y_{id})^\top $. Let $X = (X_1,...,X_n)$ and $Y = (Y_1,...,Y_n)$ be the sample matrices. In classical CCA, the covariance matrices $\Sigma_{xx}$, $\Sigma_{yy}$, and $\Sigma_{yx}$ are replaced by $\Sigma_{XX}=XX^\top /n$, $\Sigma_{YY}=YY^\top /n$, and $\Sigma_{YX}=YX^\top /n$, respectively \cite{gonzalez2008cca}. Then the estimated canonical pairs are the solutions to the following optimization problems
\begin{align}\label{estcor}
    \max_{a\in \RR^d,b\in \RR^p} & \frac{a^\top \Sigma_{YX}b}{\sqrt{a^\top \Sigma_{YY}a}\sqrt{b^\top \Sigma_{XX}b}}\nonumber\\
    {\rm s.t.} & \quad a^\top \Sigma_{YY}a_l = 0, b^\top \Sigma_{XX}b_l = 0, 1\leq l \leq k-1,
\end{align}
for $k\in\{1,...,K\}$. If $k=1$, then \eqref{estcor} becomes a unconstrained optimization problem.

If the dimension of $x$ or $y$ is larger than the sample size $n$, the classical CCA does not work because $XX^\top /n$ or $YY^\top /n$ is singular. One naive method to estimate the canonical vectors is to add diagonal matrices $\mu_Y I_d$ and $\mu_X I_p$ with $\mu_Y,\mu_X>0$ such that the estimated covariance matrix $\Sigma_{YY}+ \mu_Y I_d$ and $\Sigma_{XX}+ \mu_X I_p$ are invertible, where $I_d$ and $I_p$ are two identity matrices of size $d$ and $p$, respectively. Following the terminology in spatial statistics \cite{stein2012interpolation} and computer experiments \cite{peng2014choice}, we call $\mu_X$ and $\mu_Y$ ``nugget" parameters, and call the corresponding method CCA with a nugget parameter. CCA with a nugget parameter provides the first canonical vector $a_{\mu}$ as an eigenvector of
\begin{align}\label{ccamuforma}
     (\Sigma_{YY} + \mu_Y I_d)^{-1}\Sigma_{YX}(\Sigma_{XX} + \mu_X I_p)^{-1} \Sigma_{XY},
\end{align}
and the second canonical vector $b_{\mu}$ is proportional to 
\begin{align}\label{ccamuformb}
    (\Sigma_{XX} + \mu_X I_p)^{-1}\Sigma_{XY}a,
\end{align}
where $\Sigma_{XY} = XY^\top /n$. Although using a nugget parameter enables the matrix inverse, it may produce non-sparse canonical vectors, which may hard to interpret. Also, we are not aware of any theoretical guarantees on the consistency of estimated canonical vectors by CCA with a nugget parameter.

\subsection{Sparse canonical correlation analysis}\label{secsparsecca}
As mentioned in Section \ref{secclasscca}, if the dimension of $x$ or $y$ is larger than the sample size $n$, the classical CCA does not work because $XX^\top /n$ or $YY^\top /n$ is singular. To address the case when $p$ or $d$ is larger than $n$, many other approaches to generalize classical CCA to high dimensional settings have been proposed. In these works, thresholding or regularization is introduced into the optimization problem \eqref{CCAformpre}. For example, \cite{parkhomenko2009sparse,parkhomenko2007genome,waaijenborg2008quantifying} introduced a soft-thersholding for each element of canonical vectors. Therefore, elements with small absolute value are forced to be zero, and a sparse solution is obtained. \cite{chen2013sparse} introduced iterative thresholding to estimate the canonical vectors, and showed that the consistency of estimated canonical vectors holds under the assumptions that $\Sigma_{xx}$ and $\Sigma_{yy}$ (or the inverses of them) are sparse. \cite{tenenhaus2011regularized} proposed a regularized generalized CCA, where the constraint on canonical vectors are changed to be $\tau_1 a^\top \Sigma_{xx}a + (1-\tau_1)\|a\|_2 = 1$ and $\tau_2 b^\top \Sigma_{yy}b + (1-\tau_2)\|b\|_2 = 1$, where $\tau_1,\tau_2\in [0,1]$ are two tuning parameters.

Regularization-based sparse CCA usually estimates the $k$-th pair of canonical vectors which are obtained by solving
\begin{align}\label{spccaform}
    \max_{a,b} &\quad \frac{1}{n} a^\top YX^\top b\nonumber\\
    {\rm s.t.} & \quad\|a\|_2\leq 1, \|b\|_2\leq 1, P_1(a)\leq c_1, P_2(b) \leq c_2,\nonumber\\
    & a^\top YY^\top a_l = 0, b^\top XX^\top b_l = 0, 1\leq l \leq k-1,
\end{align}
where $P_1$ and $P_2$ are two convex penalty functions, and $c_1$ and $c_2$ are two constants. The sparsity is imposed on the canonical vectors by using different penalty functions. This method was proposed by \cite{witten2009penalized}, and has been extended by \cite{witten2009extensions}. There is no theoretical guarantees on \eqref{spccaform}, as is pointed out by \cite{chen2013sparse}. An algorithm based on \cite{witten2009extensions} has been proposed by \cite{lee2011sparse}. In \cite{waaijenborg2008quantifying}, the elastic net was also used to obtain sparsity of the estimated canonical vectors. \cite{chen2012structure} modified sparse CCA as in \cite{witten2009extensions} by adding a structure based constraint to the canonical vectors. \cite{gao2017sparse} proposed a method called convex program with group-Lasso refinement, which is a two-stage method based on group Lasso, and they proved the consistency of estimated canonical variables.

Another type of sparse CCA methods is via a reformulation approach. In \cite{hardoon2011sparse}, it was shown that based on a primal-dual framework, \eqref{estcor} with respect to the first canonical pair is equivalent to the following problem
\begin{align}\label{refcca}
    \min_{w,e}\|Y^\top w - X^\top Xe\|_2^2,
\end{align}
subject to $\|X^\top Xe\|_2^2=1$, in the sense that $(w,e)$ is the solution to \eqref{refcca} if and only if there exists $\mu,\gamma$ such that $(\mu w, \gamma Xe)$ is the solution to \eqref{estcor}. Then by imposing $l_1$ regularization on $w$ and $e$, sparse canonical vectors can be obtained. Recent work by \cite{mai2019iterative} reformulated \eqref{spccaform} into a constrained quadratic optimization problem, and proposed an iterative penalized least squares algorithm to solve the optimization problem. Theoretical guarantees on the consistency of the estimated canonical vectors were also presented in \cite{mai2019iterative}.

\subsection{Other canonical correlation analysis methods}\label{secdcca}

If one is interested in obtaining multiple canonical pairs using sparse CCA, then \eqref{spccaform} has to be solved \textit{sequentially}. In particular, in order to get $k$-th pair of canonical vectors, one must know all $l$-th pairs of canonical vectors for all $l<k$. This sequential solving problems makes sparse CCA inefficient when researchers need to estimate a relatively large number of pairs of canonical vectors. Recent works on CCA provide an alternative way, which can estimate multiple canonical pairs \textit{simultaneously}. These works include PCA-CCA \cite{song2016canonical} and decomposition-based CCA \cite{shu2020dGcca,shu2020d}. In PCA-CCA approach, a principal component analysis (PCA) rank-reduction preprocessing step is performed before applying classical CCA, yet no theoretical guarantees on this method. In decomposed-based CCA (D-CCA), the assumptions of sparsity on the canonical vectors are removed, but the random vectors $x$ and $y$ are assumed to have an ``low-rank plus noise'' model. That is, there exist random vectors $x_{r_1}\in \RR^{r_1}$ and $y_{r_2}\in\RR^{r_2}$ such that $x$ and $y$ can be written as a linear combination of $x_{r_1}$ and $y_{r_2}$ plus a noise vector. Based on this assumption, \cite{shu2020d} showed the consistency results on the canonical correlation estimators. However, unlike sparse CCA, D-CCA requires that the dimensions of both $x$ and $y$ are larger than $n$ ($\min\{p,d\}\geq cn$ for some constant $c>0$) but smaller than $n\lambda_{r_1}(\Sigma_{xx})$ and $n\lambda_{r_2}(\Sigma_{yy})$ ($\max\{p,d\}=O(n\lambda_{r_1}(\Sigma_{xx}))$ and $\max\{p,d\}=O(n\lambda_{r_2}(\Sigma_{yy}))$), where $\lambda_{i}(A)$ is the $i$-th largest eigenvalue of a matrix $A$. In particular, if $\lambda_{\max}(\Sigma_{xx})$ and $\lambda_{\max}(\Sigma_{yy})$ are bounded above by a constant (for example, $\Sigma_{xx}$ and $\Sigma_{yy}$ are identity matrices), which is a typical condition for the sparse CCA and high dimensional analysis \cite{ning2014general,gao2017sparse,chen2013sparse}, this assumption is violated. Furthermore, D-CCA does not provide sparse canonical vectors, which may be difficult to interpret.

\section{Eigenvector-based sparse canonical correlation analysis}\label{seconelasso}
In this section, we introduce the proposed method called eigenvector-based sparse canonical correlation analysis, and study its theoretical properties.

\subsection{Methodology}
In sparse canonical correlation analysis methods, the dimensions of both $x$ and $y$ can be larger than $n$, and theoretical guarantees have been provided based on the assumption that the canonical vectors are sparse. However, one needs to solve $k-1$ optimization problems \textit{sequentially} in order to obtain the $k$-th pair of canonical vectors; see \cite{mai2019iterative,waaijenborg2008quantifying} for example. These sequential algorithms cannot provide $K$ pairs of canonical vectors simultaneously. Nor do they naturally enable parallel computing \cite{jordan2013statistics,suchard2010understanding,park2012gplp}. If one is interested in estimating a relatively large number of pairs of canonical vectors, these methods may be hard to be used. On the other hand, although the D-CCA method can estimate multiple canonical pairs simultaneously, D-CCA places more restrictions on the dimensions and the covariance matrices, and does not assume sparse canonical vectors. In some cases, predictor involving all input variables is difficult to interpret, thus sparse canonical vectors are more desirable.

In this work, we consider an intermediate case, where the dimensions of $x$ and $y$ are very different. Without loss of generality for our application domain, we assume the dimension of $x$ is much larger than the sample size, while the dimension of $y$ is relatively small. We propose an \textit{eigenvector-based sparse canonical correlation analysis} (E-CCA), which can be used to estimate the canonical vectors under the the setting $p\gg n > d$. The E-CCA enjoys both advantages from the sparse CCA, where one can estimate ultrahigh-dimensional set (the dimension $p$ can increase exponentially with respect to the sample size $n$) and obtain sparse canonical vectors, and from D-CCA, where one can estimate $K>1$ pairs of canonical vectors simultaneously. Furthermore, the E-CCA naturally enables parallel computing in the algorithm, which can substantially decrease the computation time.

Unlike existing sparse CCA methods \eqref{spccaform}, we do not approach the problem directly as a ``correlation maximization" problem. First, we establish a relationship between multivariate regression and CCA, and then use this understanding to motivate our solution. This relationship allows us to apply existing methodologies from regression, which makes the algorithm of estimating canonical vectors more efficient. 
Consider a multiple linear regression on $y$ with variables $x$, 
\begin{align}\label{linearreg}
    y=B_* x+\epsilon_y,
\end{align}
where $B_*\in \RR^{d\times p}$ is the coefficient matrix. The coefficient matrix can be obtained by the projection of $y$ onto $x$. The variable $\epsilon_y$ is the projection residual, and satisfies ${\rm E}(\epsilon_y^\top B_*x) = 0$. With the relationship \eqref{linearreg}, we can compute the covariance matrices $\Sigma_{xy}$ and $\Sigma_{yy}$ by
\begin{align}\label{covarlr}
\Sigma_{xy}  = \Sigma_{xx}B_*^\top , \quad \Sigma_{yy}  = B_*\Sigma_{xx}B_*^\top  + \Sigma_{\epsilon_y\epsilon_y},
\end{align}
where the second equality follows from ${\rm E}(\epsilon_y^\top B_*x) = 0$. By the results in classical CCA, the first canonical vector of the $k$-th pair of canonical vectors $a_k$ is the $k$-th eigenvector of
\begin{align}\label{forma}
    \Sigma_{yy}^{-1}\Sigma_{yx}\Sigma_{xx}^{-1}\Sigma_{xy} & = \Sigma_{yy}^{-1}B_*\Sigma_{xx}B_*^\top ,
\end{align}
where the equality is because of \eqref{covarlr}. The second canonical vector $b_k$ is propotional to  
\begin{align}\label{formb}
    \Sigma_{xx}^{-1}\Sigma_{xy}a_k =\Sigma_{xx}^{-1}\Sigma_{xx}B_*^\top a_k = B_*^\top a_k.
\end{align}
Note that in \eqref{forma} and \eqref{formb}, we do not need to compute $\Sigma_{xx}^{-1}$. Therefore, we avoid the problem that $XX^\top /n$ is singular. Thus, if $B_*$ is known, we can replace $\Sigma_{xx}$ and $\Sigma_{yy}$ in \eqref{forma} by $XX^\top /n$ and $YY^\top /n$, respectively, to estimate the canonical vectors. 

In practice, $B_*$ is rarely known. Therefore, we need to estimate $B_*$ in order to use \eqref{forma} and \eqref{formb} to obtain the canonical vectors. Note that $B_*\in \RR^{d\times p}$ with $p\gg d$. One natural idea is to assume the coefficient matrix $B_*$ has some sparse structure, and to use the elementwise $l_1$ penalty as a regularization as in Lasso \cite{tibshirani1996regression}. If $B_*$ is sparse, the second canonical vector $b_k$ with high dimension is also sparse by \eqref{formb}. To be specific, let $\hat B$ be an estimator of $B_*$. We compute $\hat B$ by the following optimization problem
\begin{align}\label{regumodel}
    \min_{B \in \RR^{d\times p}}\sum_{i=1}^n \|Y_i - BX_i\|_2^2 + \lambda_1\|B\|_{F_1},
\end{align}
where $
    \|B\|_{F_1} = \sum_{j=1}^d \|\beta_j\|_1,
$
for $B = (\beta_1,...,\beta_d)^\top $, $\|\cdot\|_1$ is the $l_1$ norm, and $\lambda_1>0$ is a tuning parameter. Noting that \eqref{regumodel} can be rewritten as
\begin{align*}
    & \sum_{i=1}^n \|Y_i - BX_i\|_2^2 + \lambda_1\|B\|_{F_1}
    = \sum_{i=1}^n \sum_{j=1}^d (y_{ij} - \beta_j^\top X_i)^2 + \lambda_1\sum_{j=1}^d \|\beta_j\|_1  =  \sum_{j=1}^d\left(\sum_{i=1}^n  (y_{ij} - \beta_j^\top X_i)^2 + \lambda_1 \|\beta_j\|_1\right),
\end{align*}
we can decompose \eqref{regumodel} into $d$ Lasso problems,
\begin{align}\label{sapregumodel}
  \min_{\beta_j} \sum_{i=1}^n  (y_{ij} - \beta_j^\top X_i)^2 + \lambda_1 \|\beta_j\|_1
\end{align}
for $j\in\{1....,d\}$. Note that these $d$ Lasso problems are independent of each other, which allows parallel computing. Let $\hat \beta_j$ be the solution to \eqref{sapregumodel}, and $\hat B = (\hat \beta_1,...,\hat \beta_d)^\top $. Therefore, $\hat B$ is an estimator of $B_*$. By replacing $B_*$, $\Sigma_{xx}$ and $\Sigma_{yy}$ in \eqref{forma} and \eqref{formb} with $\hat B$, $XX^\top /n$ and $YY^\top /n$, respectively, we can obtain the $k$-th pair of estimated canonical vectors $\hat a_k$ and $\hat b_k$ as follows. The first estimated canonical vector $\hat a_k$ is the $k$-th eigenvector of $(YY^\top )^{-1}\hat BX(\hat BX)^\top ,$ and the second canonical vector $\hat b_k$ is proportional to $\hat B^\top \hat a_k.$ Algorithm \ref{alg:onelasso} describes the procedure to obtain the $k$-th pair of estimated canonical vectors. 

\begin{algorithm}
\caption{Eigenvector-based sparse CCA}
\begin{algorithmic}[1]
\State \textbf{Input:} Observed data $X$ and $Y$.
\State Parallelly solve \eqref{sapregumodel} for $j\in\{1....,d\}$ to obtain the estimated coefficients $\hat B$.
\State Calculate the eigenvectors of $(YY^\top )^{-1}\hat BX(\hat BX)^\top $. The $k$-th eigenvector $\hat a_k$ is the first estimated canonical vector in the $k$-th pair of estimated canonical vectors. The second estimated canonical vector is $\hat b_k'  = \hat B^\top \hat a_k$.
\State Normalize $\hat b_k'$ as $\hat b_k$ such that $\|\hat b_k\|_2 = 1$.
\State \textbf{Output: } The $k$-th pair of estimated canonical vectors $\hat a_k$ and $\hat b_k$.
\end{algorithmic}\label{alg:onelasso}
\end{algorithm}

Because $d$ is small, there is no need to impose sparsity on the first canonical vector $\hat a_k$. Similar to Lasso problem, the parameter $\lambda_1$ controls the sparsity of the estimated coefficients $\hat \beta_j$, thus the sparsity of $\hat B$. The larger $\lambda_1$ is, the more sparse $\hat B$ is. Since the second estimated canonical vector is $\hat b_k'  = \hat B^\top \hat a_k$, it can be seen that if $\hat B$ is sparse, so is $\hat b_k'$. Therefore, one can enlarge the parameter $\lambda_1$ to obtain a more sparse  $\hat b_k'$. In practice, one can use cross validation to choose the parameter $\lambda_1$. If $d$ is also large, one can apply principal component analysis (PCA) to reduce the dimension of $Y$. This approach has been applied in our real data analysis; see Section \ref{seccasestudy1}.

Clearly, Algorithm \ref{alg:onelasso} is not in the form of \eqref{spccaform} with setting $c_1$ to infinity. As we will see in Section \ref{secnumeric}, Algorithm \ref{alg:onelasso} is more efficient than \eqref{spccaform} with setting $c_1$ to infinity, i.e., not imposing penalty on the first estimated canonical vectors. In the E-CCA, the canonical vectors are derived by solving an eigenvalue problem, instead of the optimization problem \eqref{spccaform}. In Algorithm \ref{alg:onelasso}, we do not use any iteration, except in solving Lasso, which has been well studied and optimized in the literature \cite{friedman2010regularization}. Since we assume $d$ is small, the number of Lasso problems is also small. Because $d$ Lasso problems are independent with each other, one can utilize parallel computing to obtain the estimated coefficients. By solving the eigenvector problem in Step 3, we can obtain multiple pairs of estimated canonical vectors simultaneously. The parallel computing and simultaneous estimating multiple canonical pairs makes Algorithm \ref{alg:onelasso} quite efficient, as we will see in the numeric studies.

\subsection{Theoretical properties}
In this subsection, we present theoretical results of eigenvector-based sparse CCA. We mainly focus on the consistency of the estimated canonical vectors. In the rest of this work, we will use the following definitions. For notational simplicity, we will use $C,C',C_1,C_2,...$ and $K,K_1,K_2,...$ to denote the constants, of which the values can change from line to line. For two positive sequences $s_n$ and $t_n$, we write $s_n\asymp t_n$ if, for some constants $C,C'>0$, $C\leq s_n/t_n \leq C'$. Similarly, we write $s_n\gtrsim t_n$ if $s_n\geq Ct_n$ for some constant $C>0$, and $s_n\lesssim t_n$ if $s_n\leq C't_n$ for some constant $C'>0$. 

We first introduce some technical assumptions. The first assumption is the regularity conditions on the covariance matrices $\Sigma_{xx}$ and $\Sigma_{\epsilon_y\epsilon_y}$.
\begin{assumption}\label{onelassoassumcov}
Let $\lambda_{\max}(U)$ and $\lambda_{\min}(U)$ be the maximum and minimum eigenvalues of matrix $U$, respectively. Assume there exist positive constants $K_1$ and $K_2$ such that
\begin{align*}
    K_1 \leq \min \left\{\lambda_{\min}(\Sigma_{xx}), \lambda_{\min}(\Sigma_{\epsilon_y\epsilon_y}) \right\} \leq \max\{\lambda_{\max}(\Sigma_{xx}),\lambda_{\max}(\Sigma_{\epsilon_y\epsilon_y})\} \leq K_2.
\end{align*}
\end{assumption}
Assumption \ref{onelassoassumcov} assures that the eigenvalues of the covariance matrix are bounded, which is a typical condition for the high dimensional analysis; see \cite{gao2017sparse,chen2013sparse} for example. Note Assumption \ref{onelassoassumcov} does not hold in \cite{shu2020d}, where $\lambda_{\max}(\Sigma_{yy})$ is required to diverge to infinity.

The second assumption is on the coefficient matrix $B_*$.
\begin{assumption}\label{oneassumB}
Suppose $B_*$ satisfies $\sigma_{\max}(B_*) \leq K$ for some constant $K>0$, where $\sigma_{\max}(B_*)$ is the maximum singular value of $B_*$.
\end{assumption}
As a simple consequence of Assumptions \ref{onelassoassumcov} and \ref{oneassumB}, the eigenvalues of the covariance matrix $\Sigma_{yy}$ are bounded above by a constant, and bounded below from zero, as shown in the following proposition.
\begin{proposition}\label{propSigmaybound}
Suppose Assumptions \ref{onelassoassumcov} and \ref{oneassumB} hold. Then there exist positive constants $K_1$ and $K_2$ such that
\begin{align*}
    K_1 \leq \lambda_{\min}(\Sigma_{yy}) \leq \lambda_{\max}(\Sigma_{yy})\leq K_2.
\end{align*}
\end{proposition}
\begin{proof}[\textbf{\upshape Proof:}] 
Recall in \eqref{covarlr}, we have $\Sigma_{yy}  = B_*\Sigma_{xx}B_*^\top  + \Sigma_{\epsilon_y\epsilon_y}.$

By Weyl's theorem (\cite{horn2012matrix}, Theorem 4.3.1), we have 
\begin{align*}
    \lambda_{\min}(\Sigma_{yy}) \geq \lambda_{\min}(\Sigma_{\epsilon_y\epsilon_y}) + \lambda_{\min}(B_*\Sigma_{xx}B_*^\top ) \geq  \lambda_{\min}(\Sigma_{\epsilon_y\epsilon_y}) \geq K_1,
\end{align*}
where the last inequality is because of Assumption \ref{onelassoassumcov}. Using Weyl's theorem again, we can bound the largest eigenvalue $\lambda_{\max}(\Sigma_{yy})$ by
\begin{align*}
    \lambda_{\max}(\Sigma_{yy}) \leq  \lambda_{\max}(\Sigma_{\epsilon_y\epsilon_y}) + \lambda_{\max}( B_*\Sigma_{XX}B_*^\top )
    \leq \lambda_{\max}(\Sigma_{\epsilon_y\epsilon_y}) + \|B_*\|_2^2 \lambda_{\max}(\Sigma_{XX}) \leq K_2,
\end{align*}
where the last inequality is because of Assumptions \ref{onelassoassumcov} and \ref{oneassumB}. This finishes the proof.
\end{proof}
The following assumption is on the matrix $\Sigma_{yy}^{-1}B_*\Sigma_{xx}B_*^\top $.
\begin{assumption}\label{onelassoassumA}
Let $\Gamma = \Sigma_{yy}^{-1}B_*\Sigma_{xx}B_*^\top $. Suppose the Schur decomposition of $\Gamma$ with respect to $k$-th eigenvalue and eigenvector is
\begin{align*}
    Q_k^\top \Gamma Q_k = \left[
    \begin{array}{cc}
        \lambda_k &  v_k^\top \\
        0 & T_{k}
    \end{array}\right],
\end{align*}
where $Q_k = [q_k,Q_k'] \in \RR^{d\times d}$ is orthogonal (thus, $q_k$ is the $k$-th eigenvector of $\Gamma$). Assume there exist some constants $\sigma_0>0$ and $K>0$ such that for all $k\in\{1,...,d\}$, $\sigma_k = \sigma_{\min}(T_{k} - \lambda_k I)>\sigma_0$ and $\|v_k\|_2 < K$, where $\sigma_{\min}(T_{k} - \lambda_k I)$ is the minimum singular value of $T_{k} - \lambda_k I$.
\end{assumption}
Assumption \ref{onelassoassumA} imposes the conditions on the matrix $\Sigma_{yy}^{-1}B_*\Sigma_{xx}B_*^\top $, which ensures the numerical stability of the calculation of the eigenvectors of $\Sigma_{yy}^{-1}B_*\Sigma_{xx}B_*^\top $. The singular value condition $\sigma_{\min}(T_{k} - \lambda_k I)>\sigma_0$ is necessary because if $\lambda_k$ is a nondefective, repeated eigenvalue of $\Gamma$, there exist infinitely many of eigenvectors corresponding to $\lambda_k$, thus the consistency of eigenvectors cannot hold. Roughly speaking, Assumption \ref{onelassoassumA} requires that the eigenvalues of $\Sigma_{yy}^{-1}B_*\Sigma_{xx}B_*^\top $ are well separated.

The next assumption is on the tail behaviors of variables $x$, $y$, and $\epsilon_y$. 
\begin{definition}
A vector $v = (v_1,...,v_p)^\top $ is sub-Gaussian, if there exist positive constants $K$ and $\sigma$ such that $K^2({\rm E} e^{v_i^2/K^2} - 1)\leq \sigma^2$ holds for all $i\in\{1,...,p\}$.
\end{definition}

\begin{assumption}\label{subGassum}
The random variables $x$, $y$, and $\epsilon_y$ are all sub-Gaussian. Furthermore, $\epsilon_y$ is independent of $x$.
\end{assumption}
The sub-Gaussian assumption in Assumption \ref{subGassum} is also typical in high dimensional analysis. As a simple example, $x\sim N(0,\Sigma_{xx})$ and $y\sim N(0,\Sigma_{yy})$ are sub-Gaussian, where $N(0,\Sigma)$ is a multivariate normal distribution with mean zero and covariance matrix $\Sigma$. The independence assumption of $\epsilon_y$ and $x$ is slightly stronger than ${\rm E}(\epsilon_y^\top Ax) = 0$ (which can always be done by projection), and can always be satisfied by projection of $y$ onto $x$ if $x$ and $y$ are jointly normally distributed.

Under Assumptions \ref{onelassoassumcov}--\ref{subGassum}, we have the following consistency results. 
\begin{theorem}\label{thmonelasso}
Let $B_* = (\beta_1^*,...,\beta_d^*)^\top $ with $\beta_k^* = (\beta_{k1}^*,...,\beta_{kp}^*)^\top $. Suppose Assumptions \ref{onelassoassumcov}--\ref{subGassum} hold. Furthermore, assume $\max_{k}supp(\beta_k^*) = s^*$, $n^{-1/2}s^* \log p=o(1)$, and $\lambda_1\asymp\sqrt{n\log p}$, where $supp(\beta_k^*) = {\rm card}(\{j| \beta_{kj}^* \neq 0\})$ and ${\rm card}(A)$ is the cardinality of set $A$. Then with probability at least $1 - C_1d^3/p$,
\begin{align}\label{thm1ineqrate}
    \max\{\|a_k - \hat a_k\|_2, \|b_k - \hat b_k\|_2\} \lesssim & \sqrt{d(d +s^*)\log p/n},
\end{align}
for all $k\in\{1,...,d\}$, where $C_1$ is a positive constant not depending on $n$.
\end{theorem}
In Theorem \ref{thmonelasso}, it can be seen that if $d$ is small, then E-CCA can provide consistent estimators of canonical vectors, under the high dimensional settings with respect to the second random variable. Corollary \ref{Coro:1} is an immediate consequence of Theorem \ref{thmonelasso}, which shows the asymptotic results.

\begin{corollary}\label{Coro:1}
Suppose the assumptions of Theorem \ref{thmonelasso} hold. Furthermore, assume $n\rightarrow \infty$ and $d(d +s^*)\log p=o(n)$. Then $\|\hat a_k - a_k\|_2\rightarrow 0$ and $\|\hat b_k - b_k\|_2 \rightarrow 0$ for all $k=1,...,d$ with probability tending to one.
\end{corollary}

\begin{remark}
By Theorem 1 in \cite{mai2019iterative},  and the fact that $\|a - b\|_2^2 = \|a\|_2^2 - 2\|a\|_2\|b\|_2\cos(\langle a,b\rangle) + \|b\|_2^2$, it can be shown that 
\begin{align*}
    \max\left\{\|a_k - \hat a_k\|_2^2, \|b_k - \hat b_k\|_2^2\right\} = O_{{\rm Pr}}\left(s_1^2\sqrt{\frac{\log p}{n}}\right),
\end{align*}
where $s_1 = \max\{\|a_k\|_1,\|b_k\|_1\}$. Our result shows an improved convergence rate compared to \cite{mai2019iterative} in some scenarios (for example, where $n$ is large). Although we are considering a different asymptotic regime (where $d$ is small), and our theory does not cover the scenarios in \cite{mai2019iterative}, the result illustrates the difference of our approach compared to others. 
\end{remark}

\section{Numeric simulation}\label{secnumeric}
In this section, we conduct numeric studies on the applications of the eigenvector-based sparse CCA and sparse CCA methods. We compare the eigenvector-based sparse CCA (E-CCA) with CCA with a nugget parameter (nCCA) as in \eqref{ccamuforma} and \eqref{ccamuformb}, CCA with $l_1$ penalty ($l_1$-CCA) \cite{witten2009penalized}, sCCA \cite{lee2011sparse}, and rgCCA \cite{tenenhaus2011regularized,rgccapackage}. We use R packages \texttt{PMA} \cite{pmapackage}, \texttt{sCCA} \cite{sccapackage}, \texttt{RGCCA} \cite{rgccapackage} to implement $l_1$-CCA, sCCA, and rgCCA, respectively.

\subsection{Example 1}\label{sec:num1}
As a starting point, we consider the following simple example, which is inspired by the numeric examples in \cite{witten2009penalized,witten2009extensions}. Consider two random variables
\begin{align}\label{numeg1}
    X_1 = A_1u + \epsilon_1,\quad X_2 = A_2u + \epsilon_2,
\end{align}
where $A_1\in \RR^{30\times 50}$ and $A_2\in \RR^{1000\times 50}$ are fixed matrix, and $u\in \RR^{50}$ is a random variable where each element of $u\in \RR^{50}$ is uniformly distributed on $(-0.5,0.5)$. We generate each element in $A_1\in \RR^{30\times 50}$ by Unif(0,2). We generate the sparse matrix $A_2\in \RR^{1000\times 50}$ by the following rule. For each column in $A_2$, we randomly select 50 elements, and generate each element in these 50 elements by Unif(0,2); the other 950 elements are set to be zero. The $\epsilon_1$ and $\epsilon_2$ are normally distributed random variables, with mean zero and variance 0.1. We use \eqref{CCAformpre} to compute the maximum canonical correlation, which is very close to one. In this numeric simulation example, we only consider the first canonical pair. More complicated examples with multiple canonical pairs have been considered in Section \ref{sec:num2}. 

E-CCA, $l_1$-CCA and rgCCA can provide a canonical pair in 0.2 second, while sCCA needs about 4 seconds. nCCA needs to solve a matrix inversion with size $1000\times 1000$, which is too time consuming. Therefore, we only compare the performance of E-CCA, $l_1$-CCA, rgCCA, and sCCA, and do not consider nCCA. We run 50 replicates. For each replicate, we sample the training data and the test data from the true distribution \eqref{numeg1}. Both training data set and the test data set have size 50. We estimate the canonical correlation using the training data set, and use the test data set to compute the canonical correlation. All methods provide negative correlation on the test data set sometimes. E-CCA, $l_1$-CCA, sCCA, and rgCCA provide negative correlations for 4, 23, 23, and 28 times, respectively. Even for the positive correlations,  E-CCA can provide a higher correlation on the test set. The boxplots for all canonical correlations and positive canonical correlations are shown in Figure \ref{figeg1} (a) and (b), respectively. It can be seen that in this simple example, our method can estimate the canonical correlation more accurately than other competing methods.

\begin{figure}[!ht]
\centering
\begin{subfigure}[b]{0.4\textwidth}
    \centering
    \includegraphics[width=\textwidth,height = \textwidth]{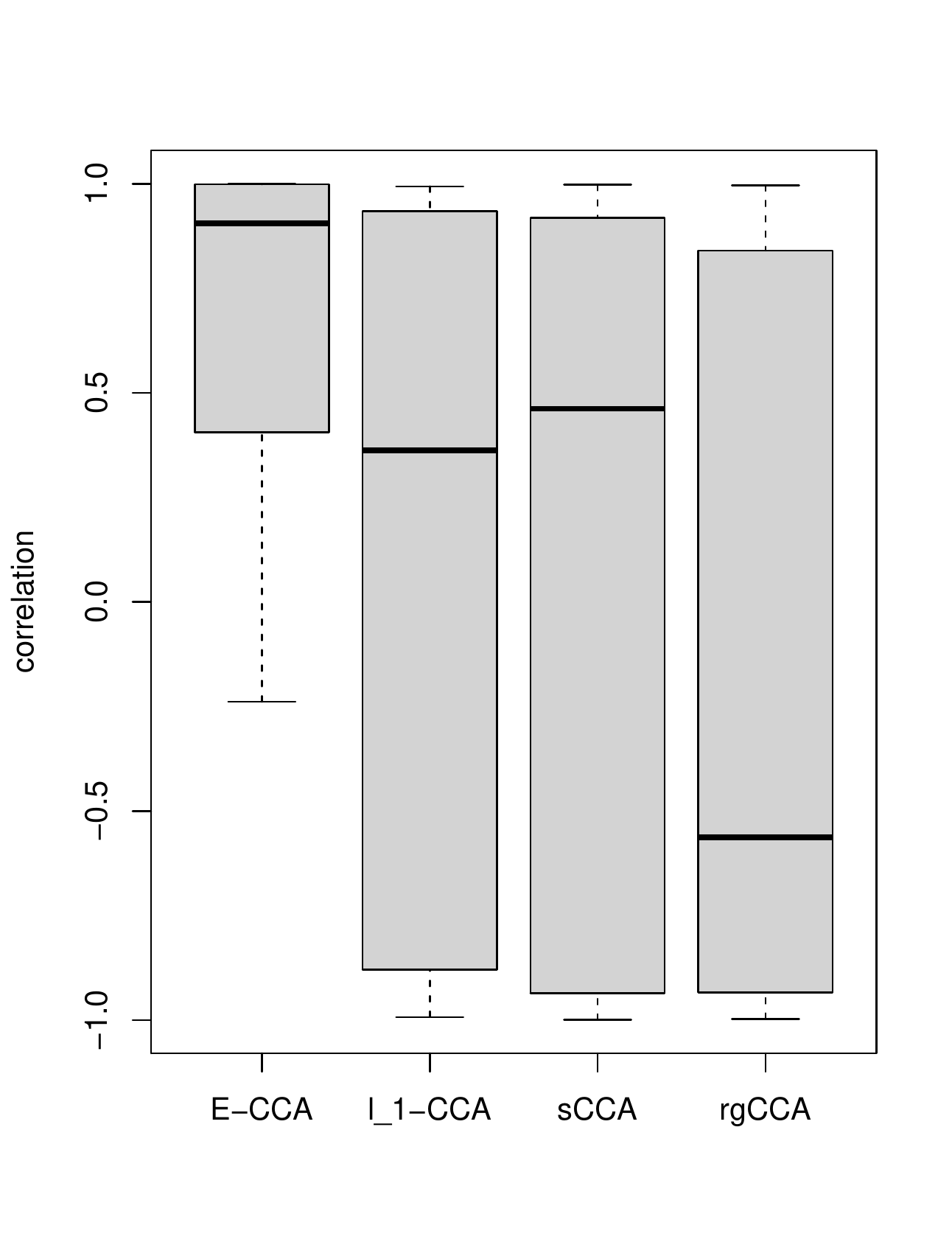}
    \caption{All canonical correlations.}
\end{subfigure}
    \begin{subfigure}[b]{0.4\textwidth}
    \centering
    \includegraphics[width=\textwidth,height = \textwidth]{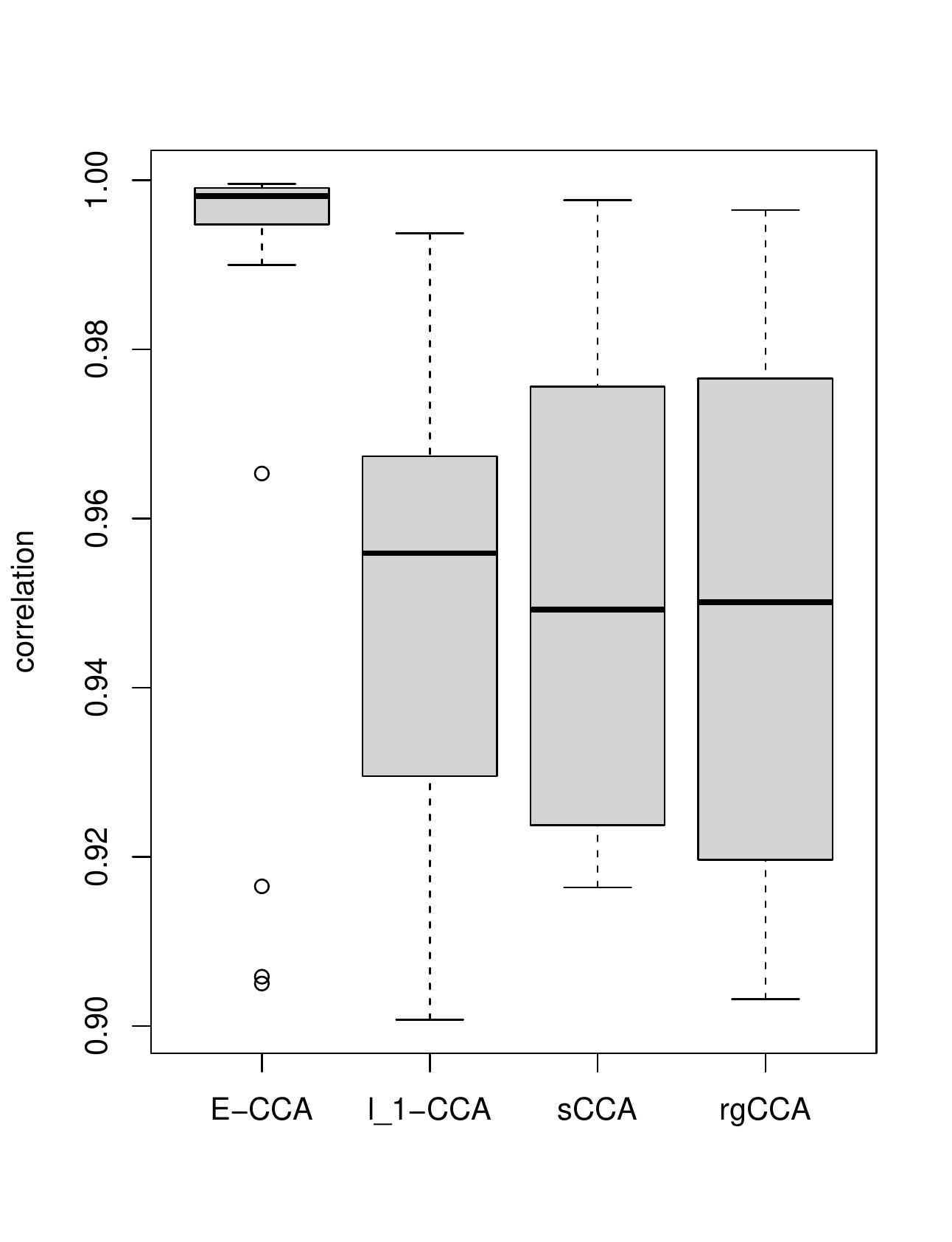}
    \caption{All positive canonical correlations.}
\end{subfigure}
    \caption{The canonical correlations for E-CCA, $l_1$-CCA, sCCA, and rgCCA in Example 1.}\label{figeg1}
\end{figure}

\subsection{Example 2}\label{sec:num2}
In this simulation study, we first simulate $X \sim N(0,\Sigma_{xx})$, $Y\sim N(0,\Sigma_{yy})$, with 
\begin{align}\label{eqnumcov}
    \Sigma_{xx} & = A_1A_1^\top  + 0.1I_p,\Sigma_{yy} = B_*\Sigma_{xx}B_*^\top  + \sigma^2I_d, \Sigma_{xy}  = \Sigma_{xx}B_*^\top ,
\end{align}
where $B_*\in \RR^{d\times p}$ is a sparse matrix, $A_1\in \RR^{p\times d}$, $\sigma>0$ is a parameter, and $I_p$ and $I_d$ are identity matrices with size $p$ and $d$, respectively.

Given $\Sigma_{xx}$, $B_*$ and $\sigma^2$, we can calculate the $k$-th true canonical vectors by \eqref{traforma} and \eqref{traformb}, denoted by $a_k$ and $b_k$, respectively. In the numeric simulation, we mainly focus on the first three pairs of canonical vectors $(a_1,b_1)$, $(a_2,b_2)$, and $(a_3,b_3)$. 

For E-CCA and nCCA, we generate an independent validation set of $X_v$ and $Y_v$ with the same sample size as the training set. The validation set is used to select the tuning parameters. Specifically, let $\lambda_1,...,\lambda_m$ be candidates of tuning parameters, and $(\hat a_{1,1}', \hat b_{1,1}'),...,(\hat a_{1,m}',\hat b_{1,m}')$ be the canonical vectors obtained by using parameters $\lambda_1,...,\lambda_m$, respectively. Then we compute Corr$(Y_v^\top a_{1,j}',X_v^\top b_{1,j}')$ for $j=1,...,m$, and choose $k= \argmax_{1\leq j\leq m}{\rm Corr}(Y_v^\top a_{1,j}',X_v^\top b_{1,j}')$. The tuning parameter then is chosen to be $\lambda_k$, and the estimated canonical vectors are $\hat a_i = \hat a_{i,k}'$ and $\hat b_i =\hat b_{i,k}'$. In $l_1$-CCA, since the dimension of $y$ is less than the sample size, we do not impose penalty on the canonical vectors $a_i$, and use six candidates of tuning parameters for the penalty term on $b_i$. For E-CCA and nCCA, we also use six candidates of tuning parameters. We use default settings for sCCA. After obtaining estimated canonical vectors, we compare the $l_2$ errors $(\sum_{i=1}^3\|\hat a_i - a_i\|_2^2)^{1/2}$ and $(\sum_{i=1}^3\|\hat b_i - b_i\|_2^2)^{1/2}$ for all four methods.

Note in \eqref{eqnumcov}, the parameter $\sigma^2$ controls the correlation between $x$ and $y$. Roughly speaking, a larger $\sigma^2$ leads to a smaller correlation between $x$ and $y$. Therefore, we choose $\sigma^2 = 0.1k$, for $k\in\{3,...,24\}$ (when $(n,p,d)=(1500,2500,10)$ in Case 1, we choose $\sigma^2 = 0.1k$, for $k\in\{3,...,7\}$, because the computation time becomes much larger) to see the change of $l_2$ errors when the correlation of $x$ and $y$ changes. For each $k$, we run $N=50$ replicates (when $(n,p,d)=(1500,2500,10)$ in Case 1, we set $N=25$). For $j$-th replicate, we compute the estimated canonical vectors $\hat a_{i,j}$ and $\hat b_{i,j}$, and use 
\begin{align*}
     \left(\hat{{\rm E}}\sum_{i=1}^3\|\hat a_i - a_i\|_2^2\right)^{1/2} = \left(\frac{1}{N}\sum_{j=1}^N \sum_{i=1}^3\|\hat a_{i,j} - a_i\|_2^2\right)^{1/2}, \left(\hat{{\rm E}}\|\hat b_i - b_i\|_2^2\right)^{1/2}=\left(\frac{1}{N}\sum_{j=1}^N \sum_{i=1}^3\|\hat b_{i,j} - b_i\|_2^2\right)^{1/2}
\end{align*}
to approximate the root mean squared prediction error (RMSE) $({\rm E}\sum_{i=1}^3\|\hat a_i - a_i\|_2^2)^{1/2}, ({\rm E}\sum_{i=1}^3\|\hat b_i - b_i\|_2^2)^{1/2},$ respectively. We also collect the computation time of these four methods.

We consider two cases, where the matrix $B_*$ is different. In both cases, we use the sample size $n=500$. The matrix $A_1 = (\alpha_{jk})_{jk}$ is randomly generated by
\begin{align*}
    \alpha_{jk} \left\{ \begin{array}{cc}
    \sim {\rm Unif}(0,2) & \mbox{ with probability 0.3,}\\
     = 0 & \mbox{ with probability 0.7,}
    \end{array}\right.
\end{align*}
where $ {\rm Unif}(0,2)$ is the uniform distribution on the interval $[0,2]$. 

\textbf{Case 1:} In \eqref{eqnumcov}, we choose $B_* = (B_1,B_2)^\top $, where $B_1=(B^{(1)}_{ij})\in \RR^{d\times (d+1)}$ with $B^{(1)}_{ii}=1$, $B^{(1)}_{i,i-1}=0.4, B^{(1)}_{i,i+1}=0.2, B^{(1)}_{i,i-2}=0.1$ and all other elements zero, and $B_2\in \RR^{d\times (p-d-1)}$ is a zero matrix. The results of approximated root mean squared prediction errors $\left(\hat{{\rm E}}\sum_{i=1}^3\|\hat a_i - a_i\|_2^2\right)^{1/2}$, $\left(\hat{{\rm E}}\sum_{i=1}^3\|\hat b_i - b_i\|_2^2\right)^{1/2}$ and $\left(\hat{{\rm E}}\sum_{i=1}^3\|\hat a_i - a_i\|_2^2\right)^{1/2}+\left(\hat{{\rm E}}\sum_{i=1}^3\|\hat b_i - b_i\|_2^2\right)^{1/2}$, and the computation time for one replicate are shown in Figure \ref{fig:case1ab}. 

\textbf{Case 2:}
We choose $B_* = (B_1,B_2)^\top $ in \eqref{eqnumcov}, where $B_2=(B^{(2)}_{ij})\in \RR^{d\times d}$ with $B^{(2)}_{i,d-i+1}=2$, $B^{(2)}_{i,d-i+2}=B^{(2)}_{i,d-i}=1$ and all other elements zero, and $B_1\in \RR^{d\times (p-d)}$ is a zero matrix. The results of approximated root mean squared prediction errors $\left(\hat{{\rm E}}\sum_{i=1}^3\|\hat a_i - a_i\|_2^2\right)^{1/2}$, $\left(\hat{{\rm E}}\sum_{i=1}^3\|\hat b_i - b_i\|_2^2\right)^{1/2}$ and $\left(\hat{{\rm E}}\sum_{i=1}^3\|\hat a_i - a_i\|_2^2\right)^{1/2}+\left(\hat{{\rm E}}\sum_{i=1}^3\|\hat b_i - b_i\|_2^2\right)^{1/2}$, and the computation time for one replicate are shown in Figure \ref{fig:case2ab}. Note that since all methods perform poorly when $(n,p,d)=(1500,2500,10)$ in Case 2, we omit the results for that case.

\begin{figure}[!ht]
    \centering
    \begin{subfigure}[b]{0.24\textwidth}
        \centering
        \includegraphics[width=\textwidth]{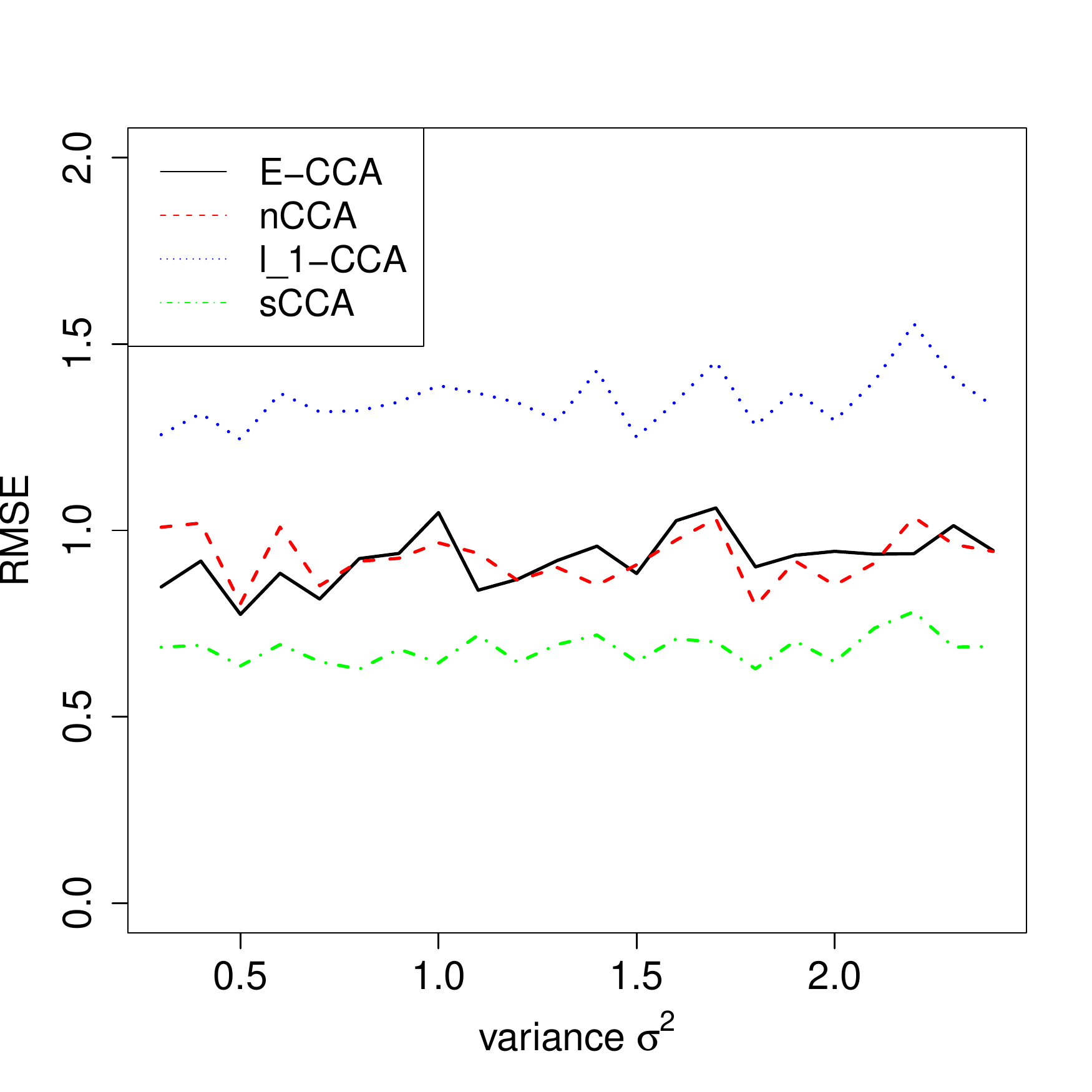}
        \caption{$(\hat{{\rm E}}\sum_{i=1}^3\|\hat a_i - a_i\|_2^2)^{1/2}$.}
    \end{subfigure}
    \begin{subfigure}[b]{0.24\textwidth}
        \centering
        \includegraphics[width=\textwidth]{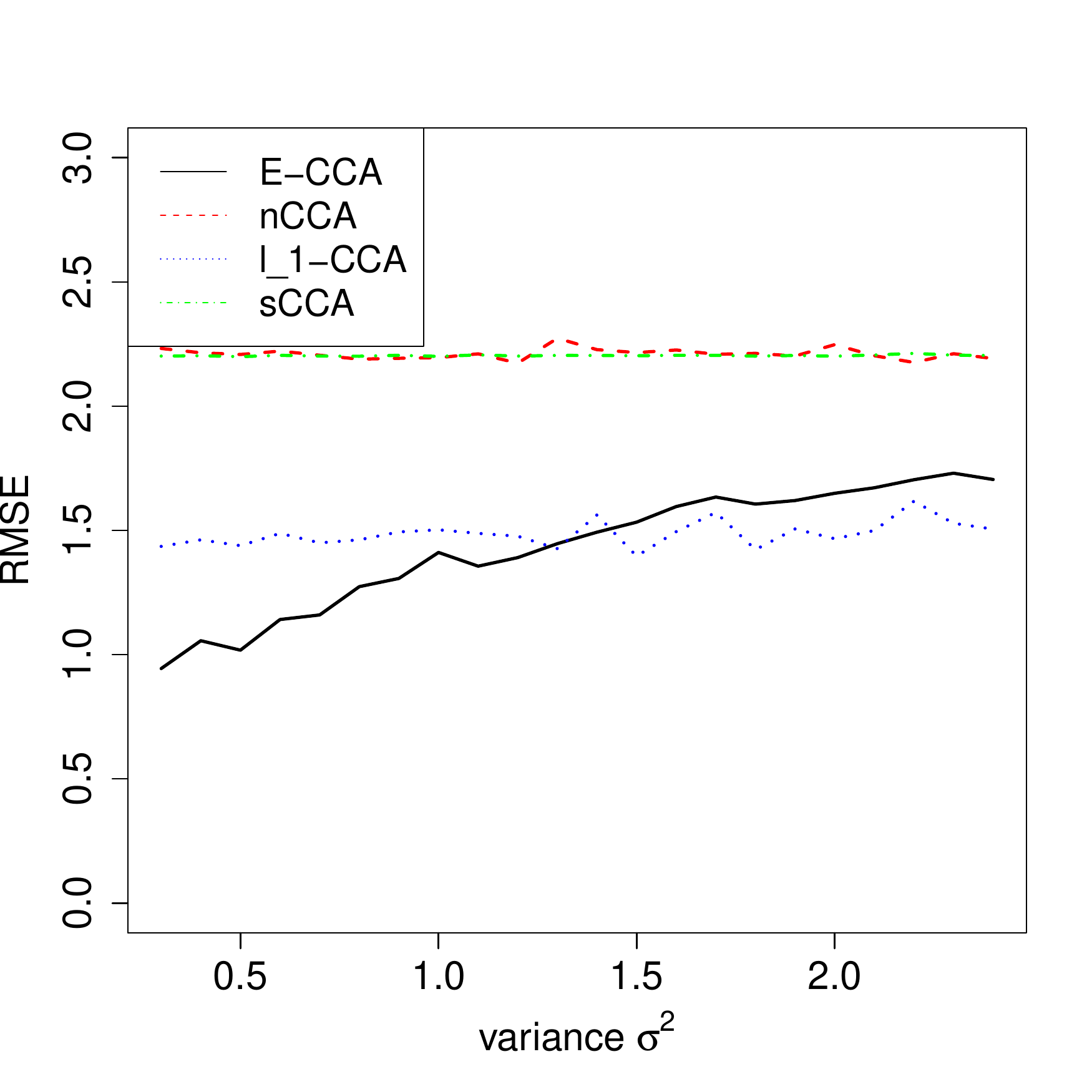}
        \caption{$(\hat{{\rm E}}\sum_{i=1}^3\|\hat b_i - b_i\|_2^2)^{1/2}$.}
    \end{subfigure}
    \begin{subfigure}[b]{0.24\textwidth}
        \centering
        \includegraphics[width=\textwidth]{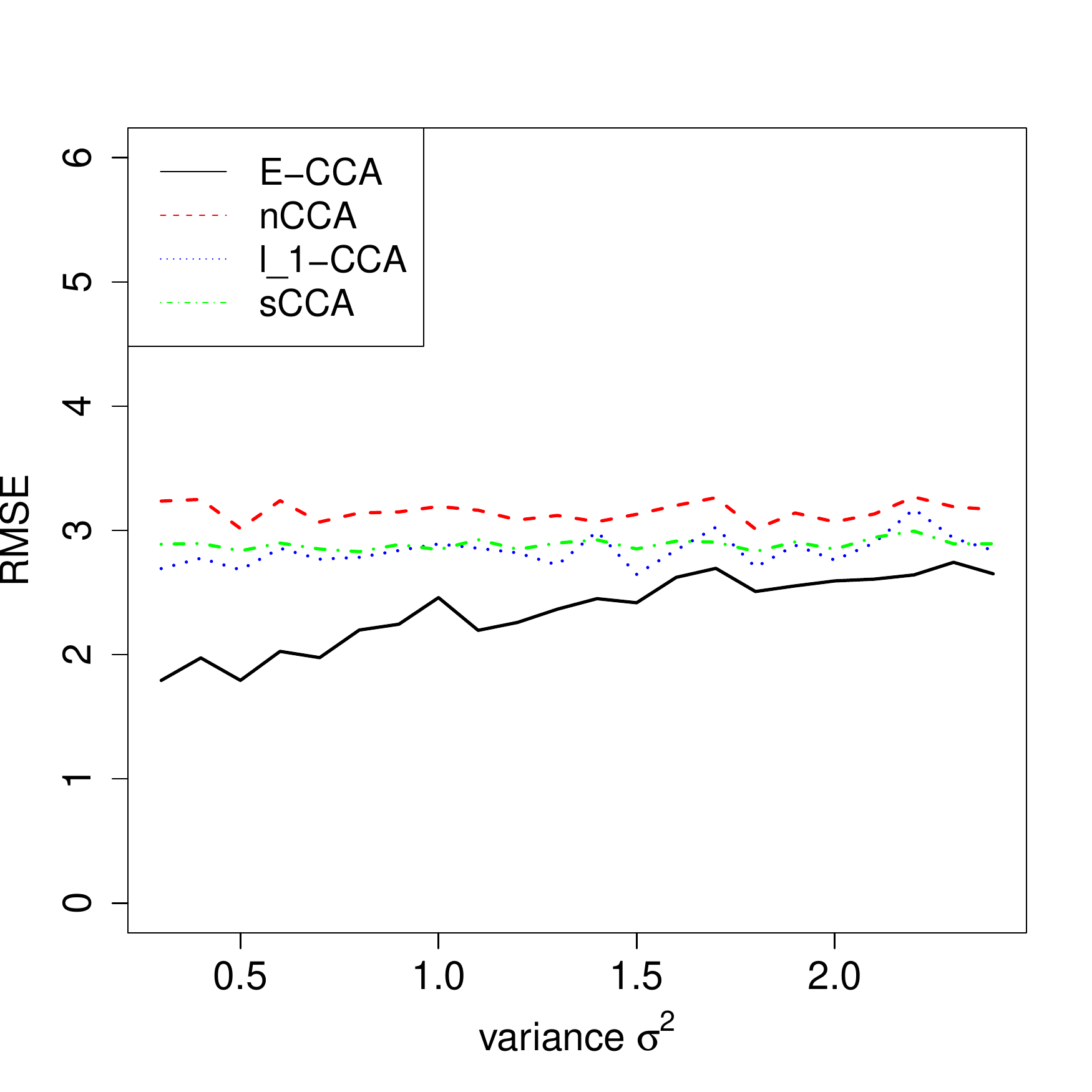}
        \caption{Total error.}
    \end{subfigure}
    \begin{subfigure}[b]{0.24\textwidth}
        \centering
        \includegraphics[width=\textwidth]{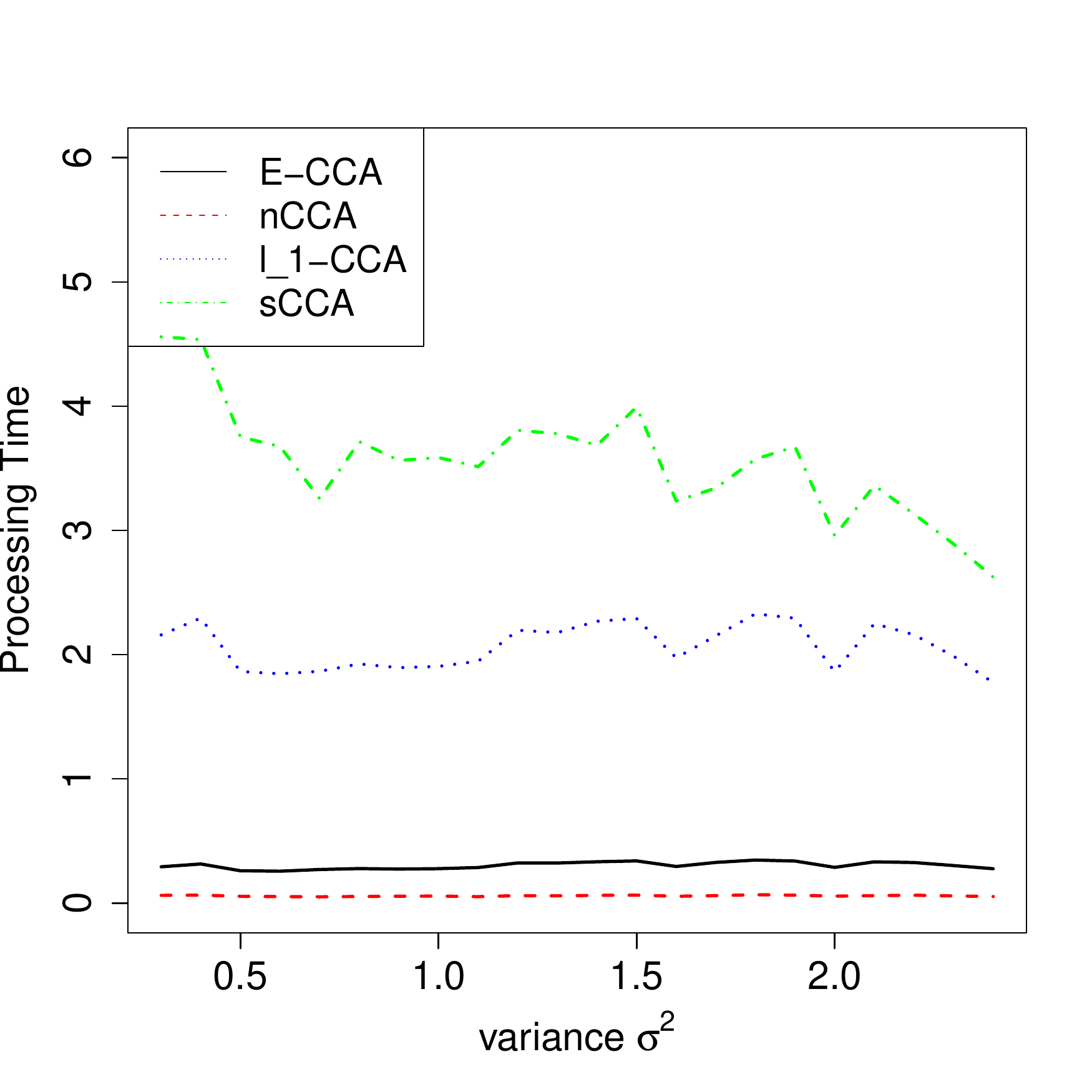}
        \caption{Processing time.}
    \end{subfigure}
    \begin{subfigure}[b]{0.24\textwidth}
        \centering
        \includegraphics[width=\textwidth]{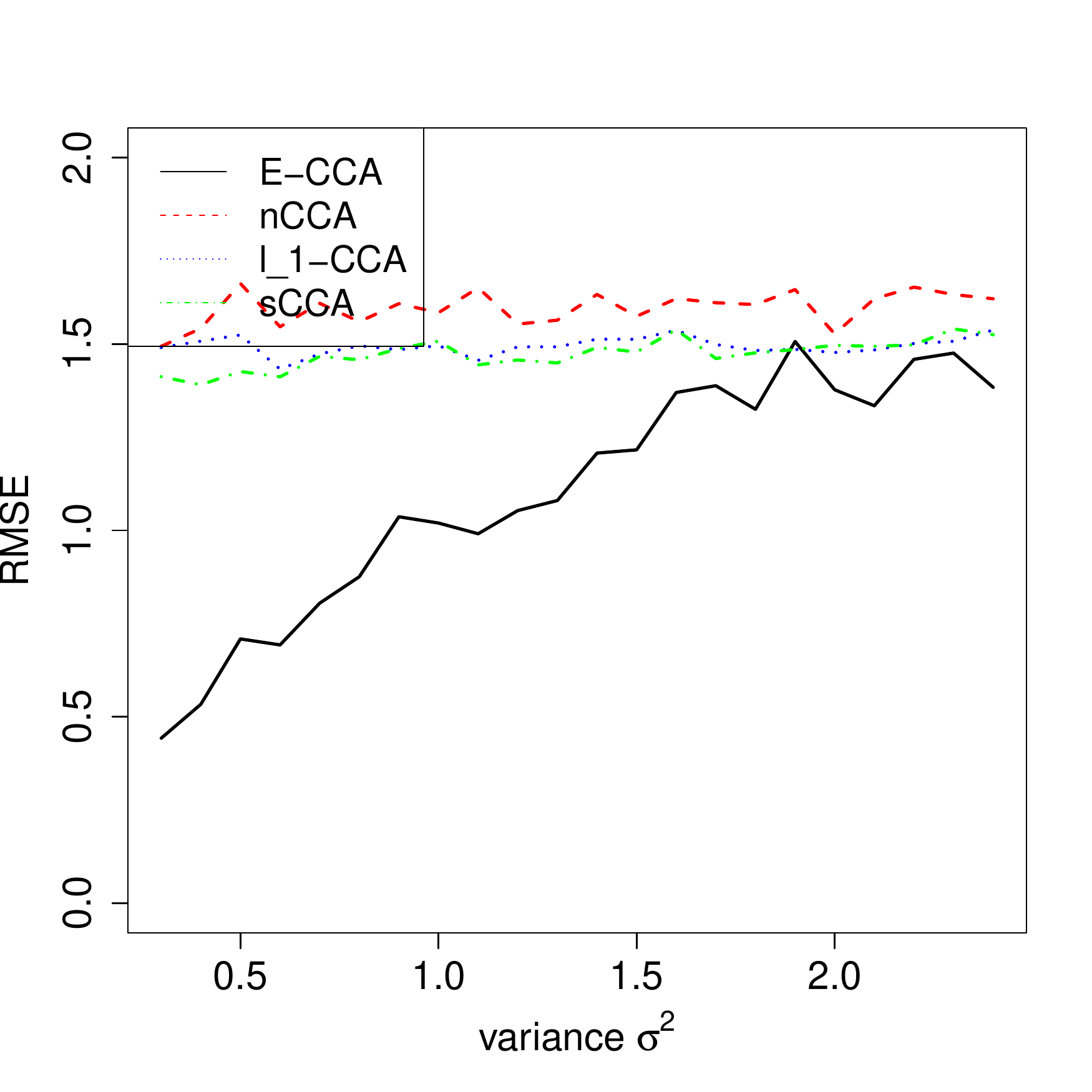}
        \caption{$(\sum_{i=1}^3\hat{{\rm E}}\|\hat a_i - a_i\|_2^2)^{1/2}$.}
    \end{subfigure}
        \begin{subfigure}[b]{0.24\textwidth}
        \centering
        \includegraphics[width=\textwidth]{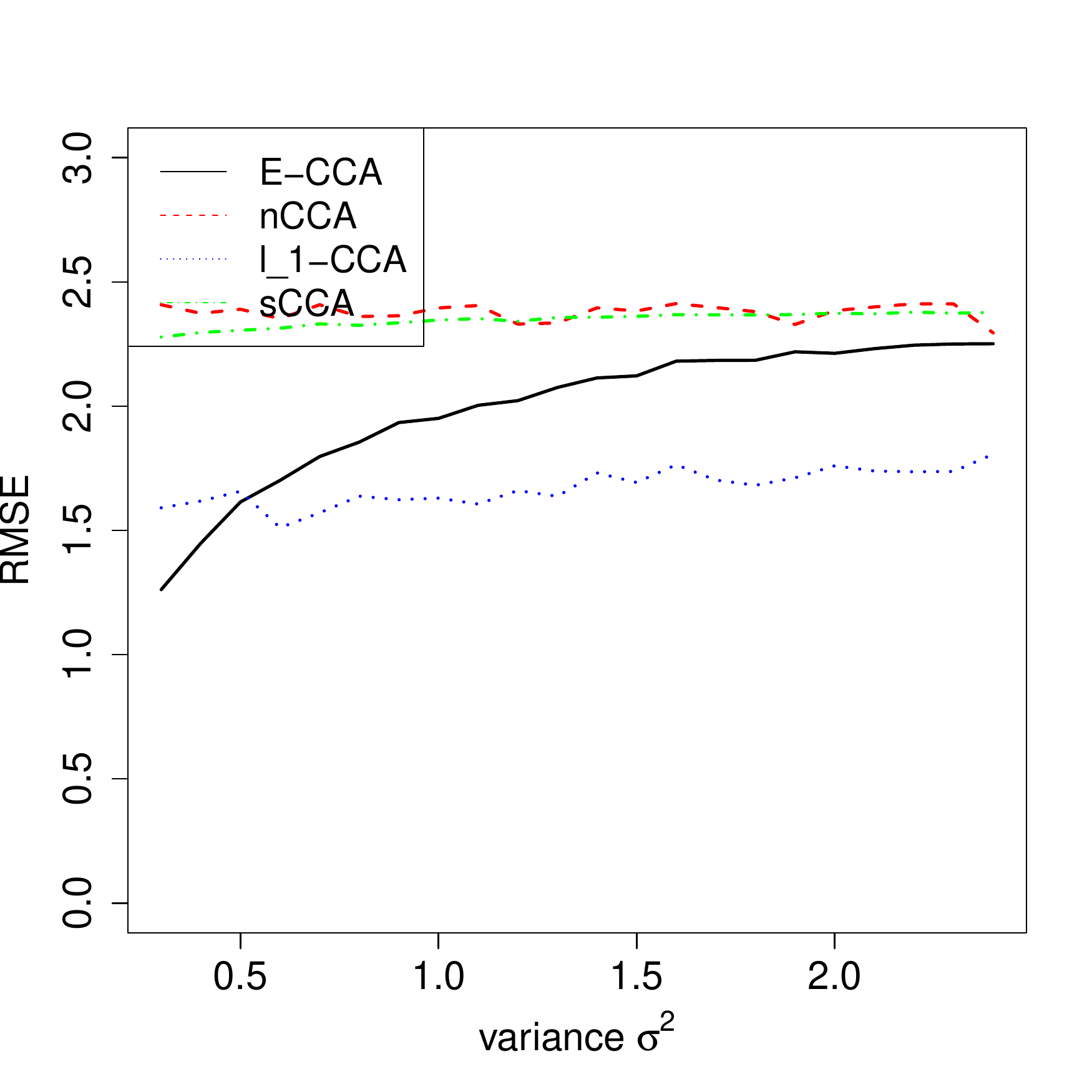}
        \caption{$(\sum_{i=1}^3\hat{{\rm E}}\|\hat b_i - b_i\|_2^2)^{1/2}$.}
    \end{subfigure}
    \begin{subfigure}[b]{0.24\textwidth}
        \centering
        \includegraphics[width=\textwidth]{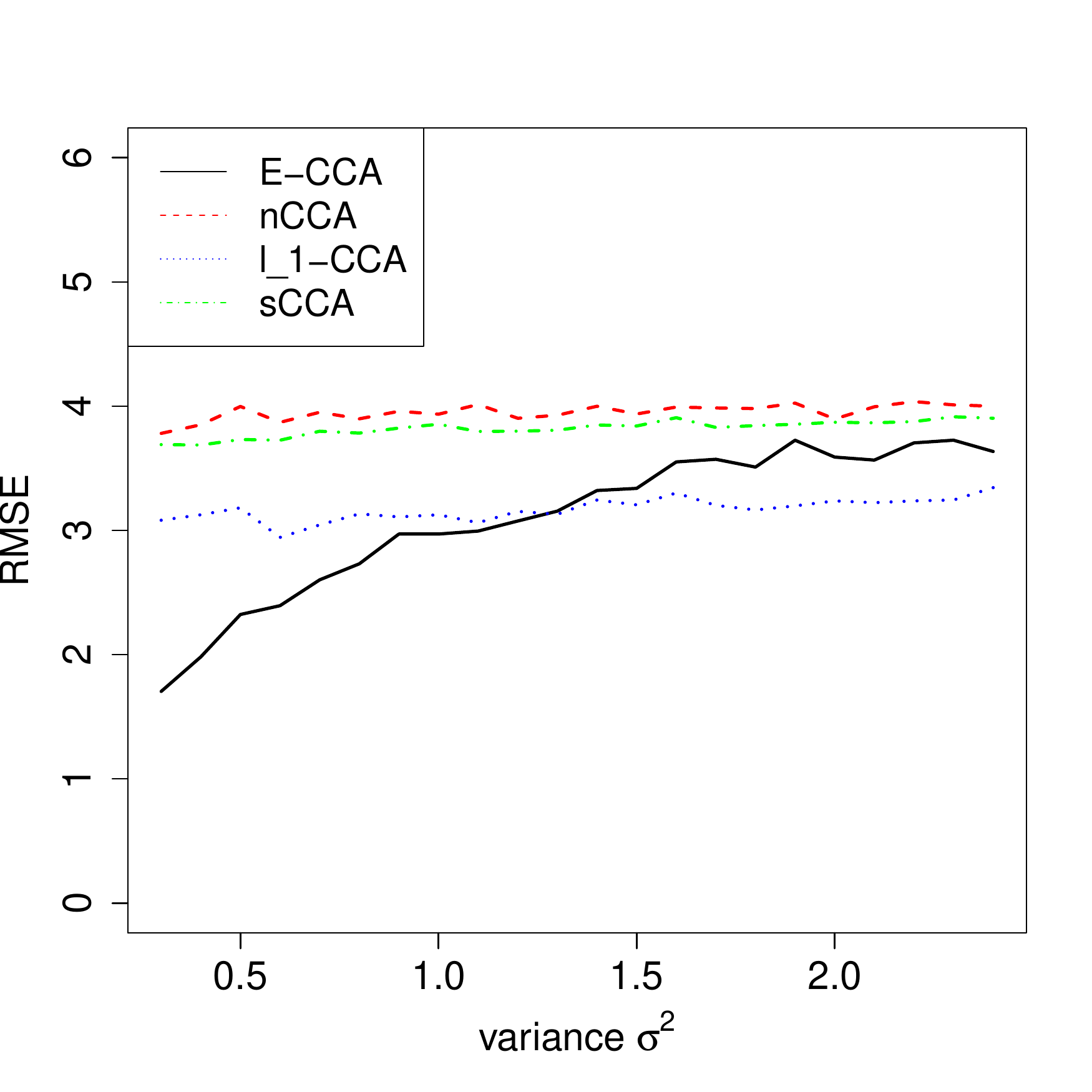}
        \caption{Total error.}
    \end{subfigure}
    \begin{subfigure}[b]{0.24\textwidth}
        \centering
        \includegraphics[width=\textwidth]{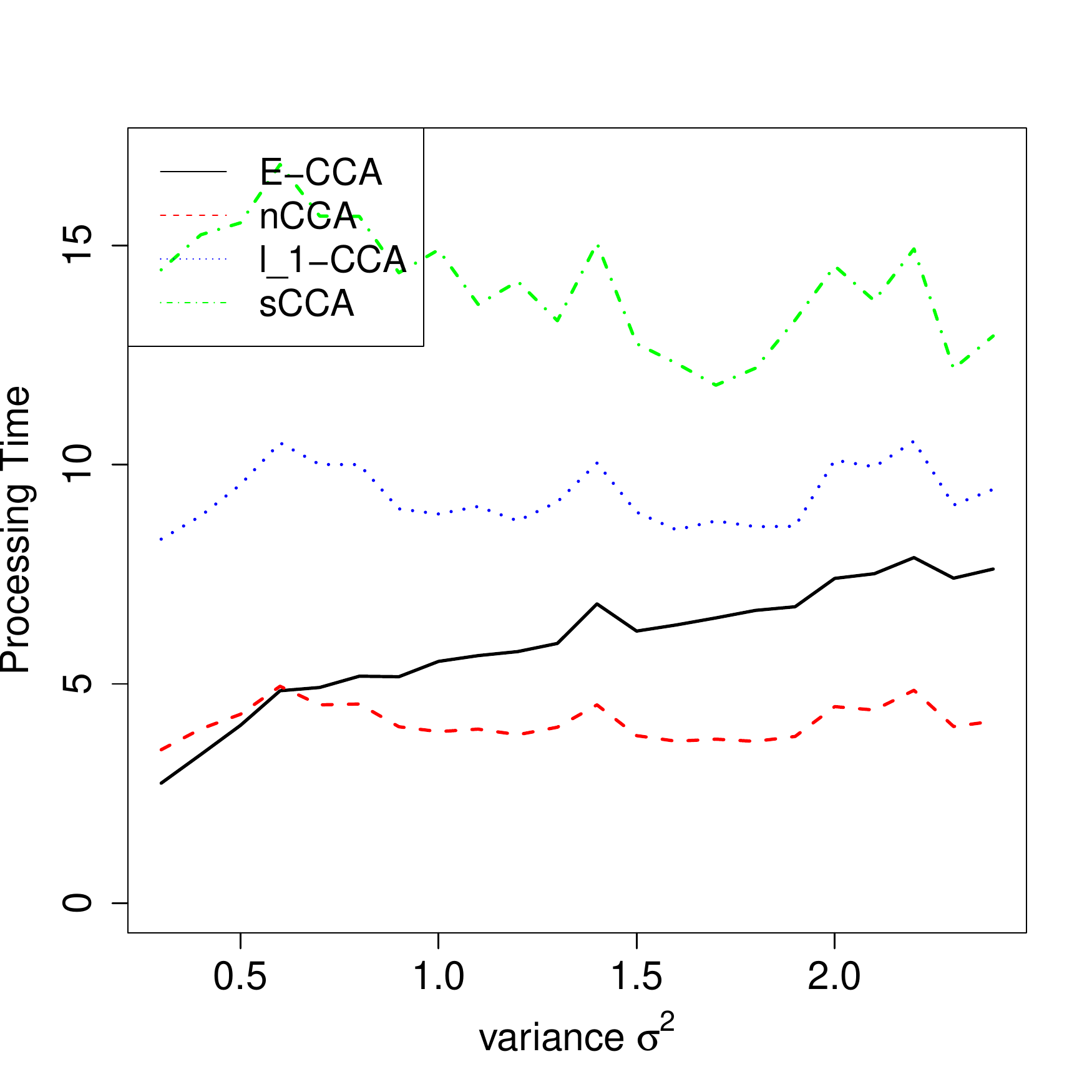}
        \caption{Processing time.}
    \end{subfigure}
    
    \begin{subfigure}[b]{0.24\textwidth}
        \centering
        \includegraphics[width=\textwidth]{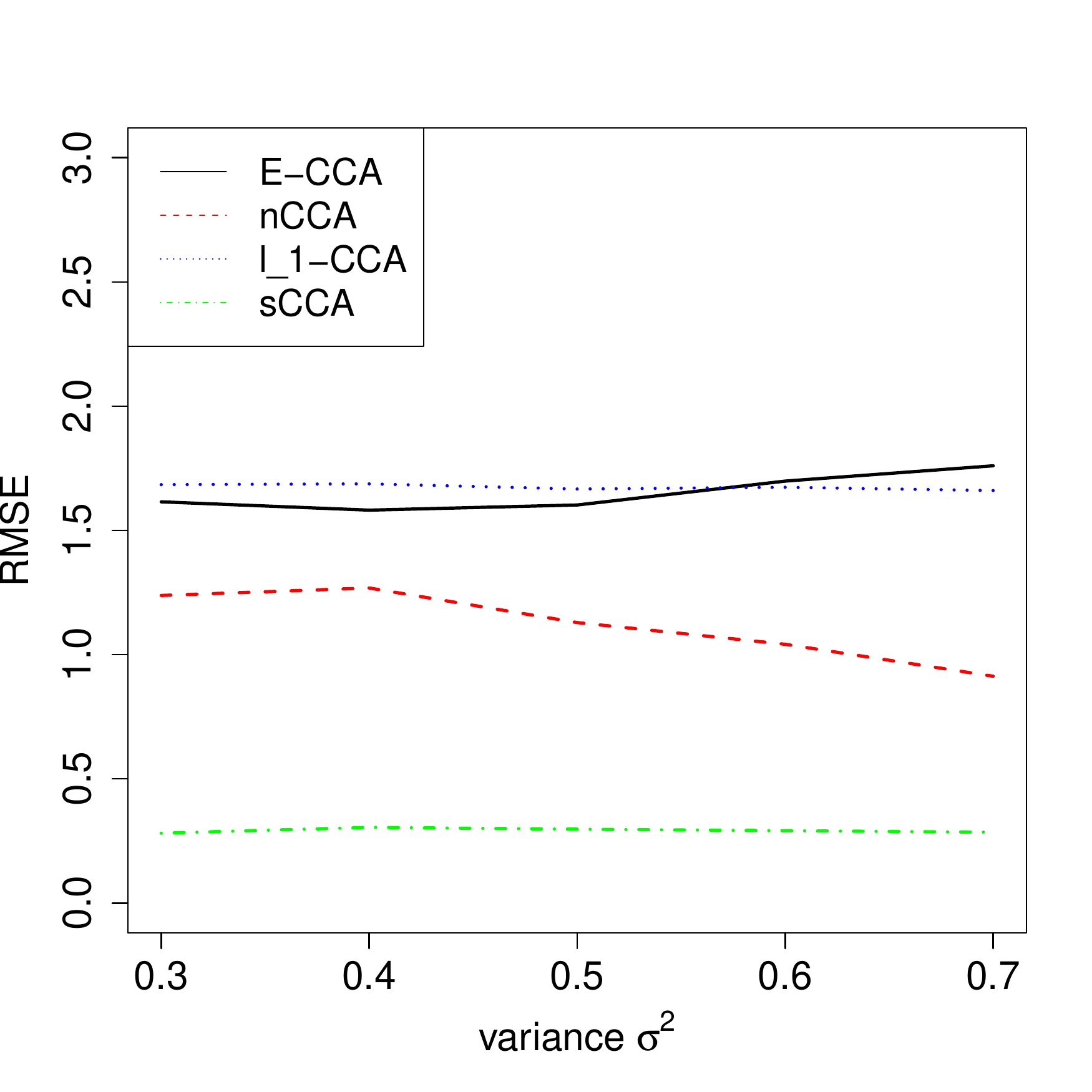}
        \caption{$(\sum_{i=1}^3\hat{{\rm E}}\|\hat a_i - a_i\|_2^2)^{1/2}$.}
    \end{subfigure}
    \begin{subfigure}[b]{0.24\textwidth}
        \centering
        \includegraphics[width=\textwidth]{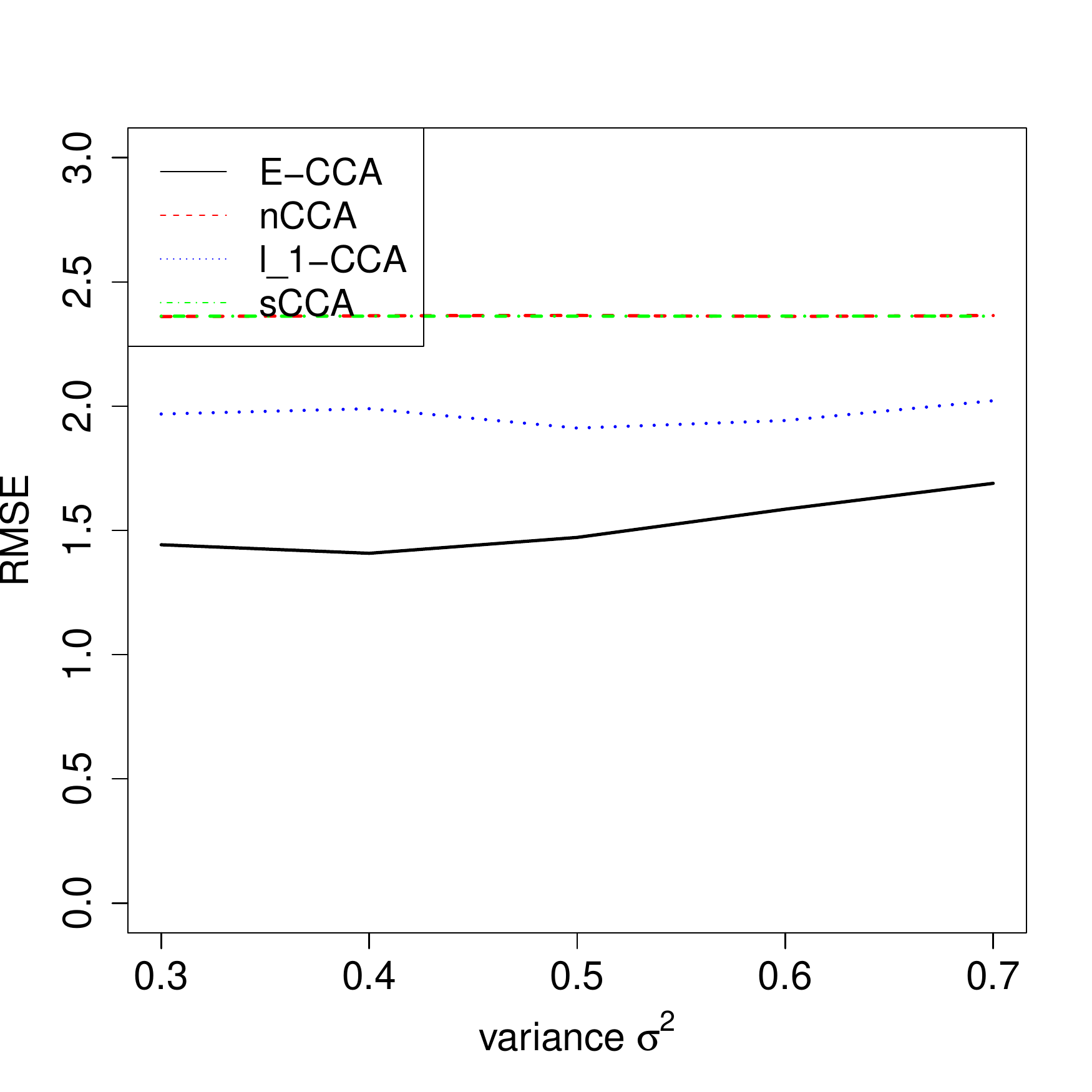}
        \caption{$(\sum_{i=1}^3\hat{{\rm E}}\|\hat b_i - b_i\|_2^2)^{1/2}$.}
    \end{subfigure}
    \begin{subfigure}[b]{0.24\textwidth}
        \centering
        \includegraphics[width=\textwidth]{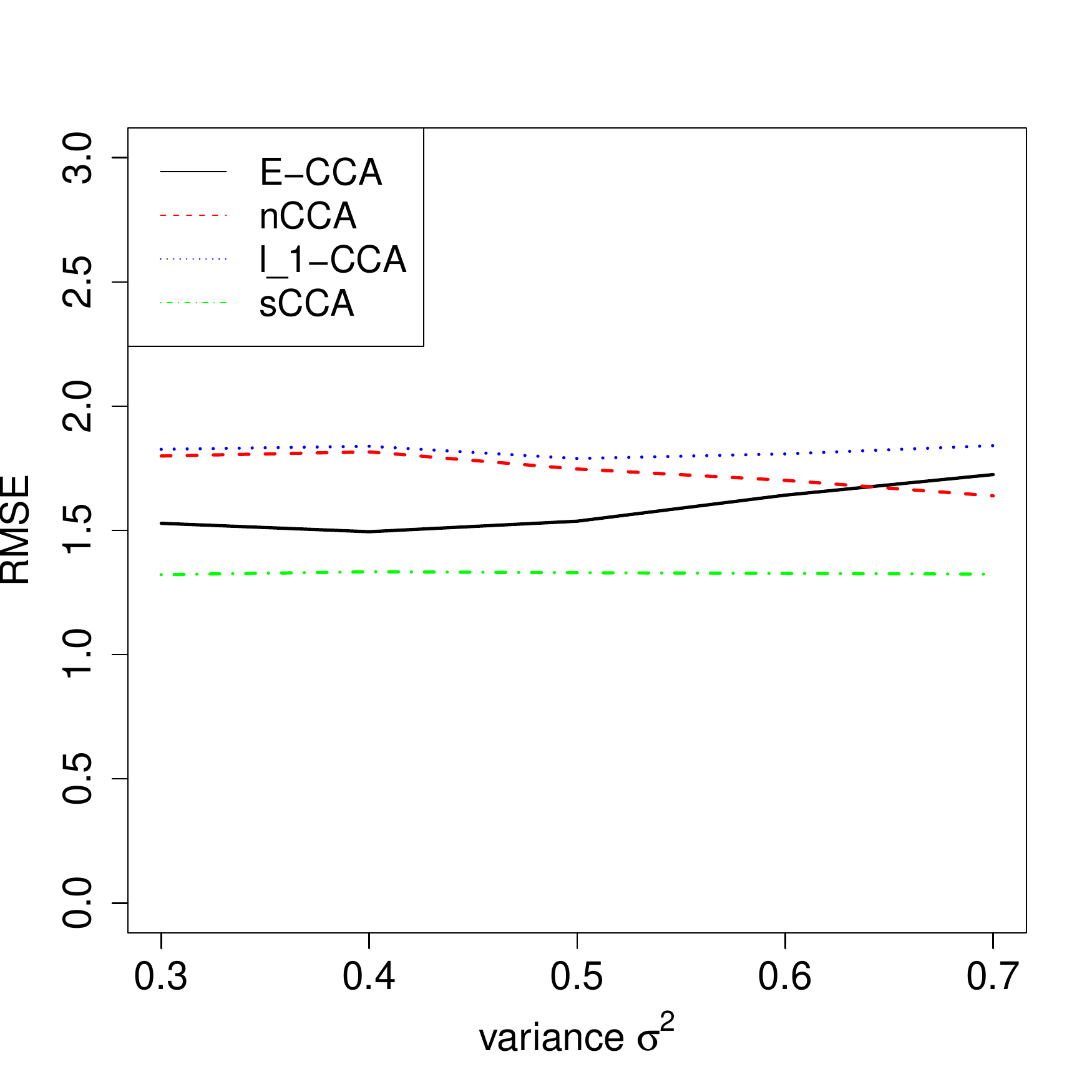}
        \caption{Total error.}
    \end{subfigure}
    \begin{subfigure}[b]{0.24\textwidth}
        \centering
        \includegraphics[width=\textwidth]{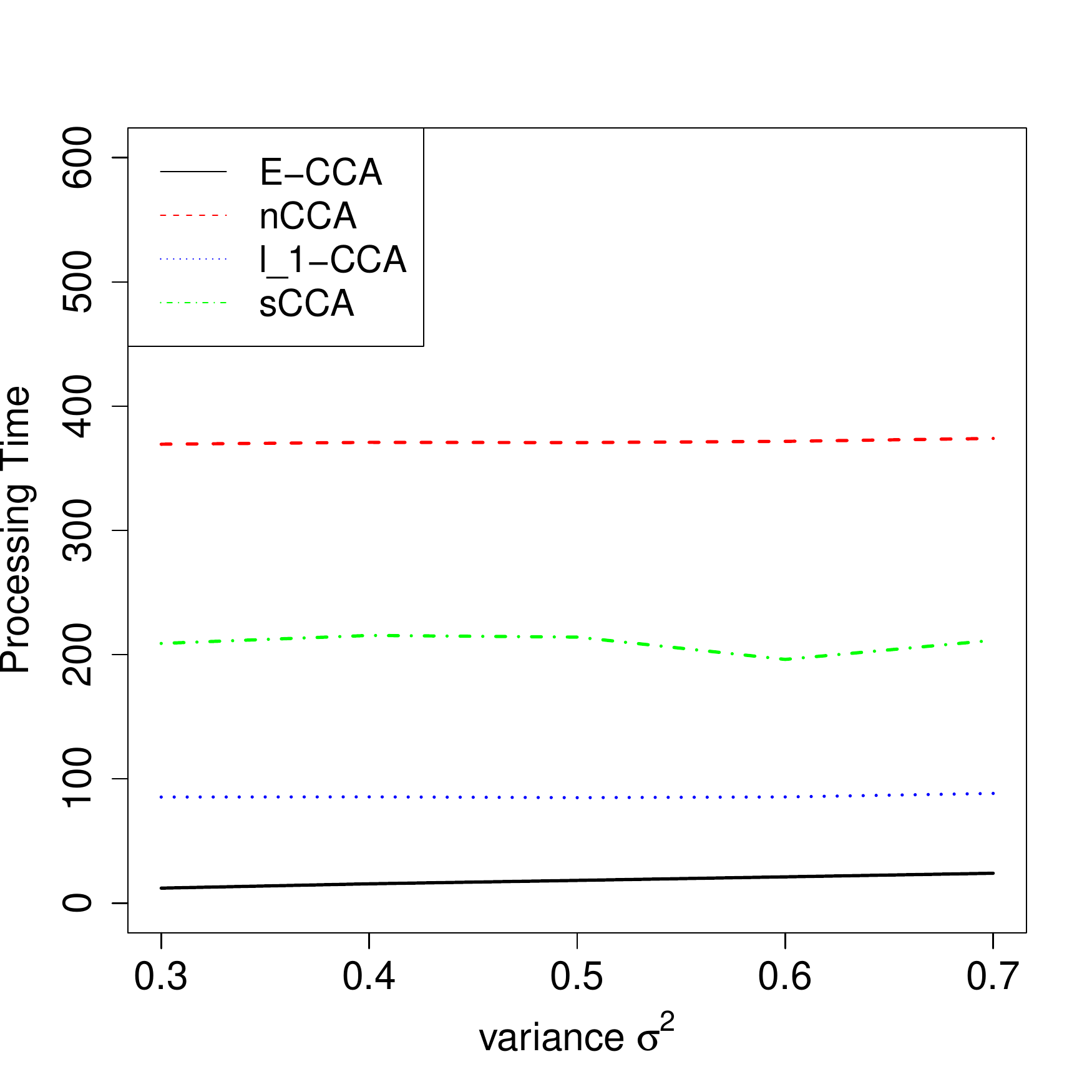}
        \caption{Processing time.}
    \end{subfigure}
    \caption{The average of approximated root mean squared prediction errors and processing time for Case 1. The processing time is in seconds. In subfigures (c) and (g), Total error = $(\hat{{\rm E}}\sum_{i=1}^3\|\hat a_i - a_i\|_2^2)^{1/2} +(\hat{{\rm E}}\sum_{i=1}^3\|\hat b_i - b_i\|_2^2)^{1/2}$. \textbf{Row 1:} $(p,d)=(100,5)$. \textbf{Row 2:} $(p,d)=(500,3)$.  \textbf{Row 3:} $(n,p,d)=(1500,2500,10)$. Note in Row 3, the variance is from 0.3 to 0.7. }
    \label{fig:case1ab}
\end{figure}

\begin{figure}[!ht]
    \centering
    \begin{subfigure}[b]{0.24\textwidth}
        \centering
        \includegraphics[width=\textwidth]{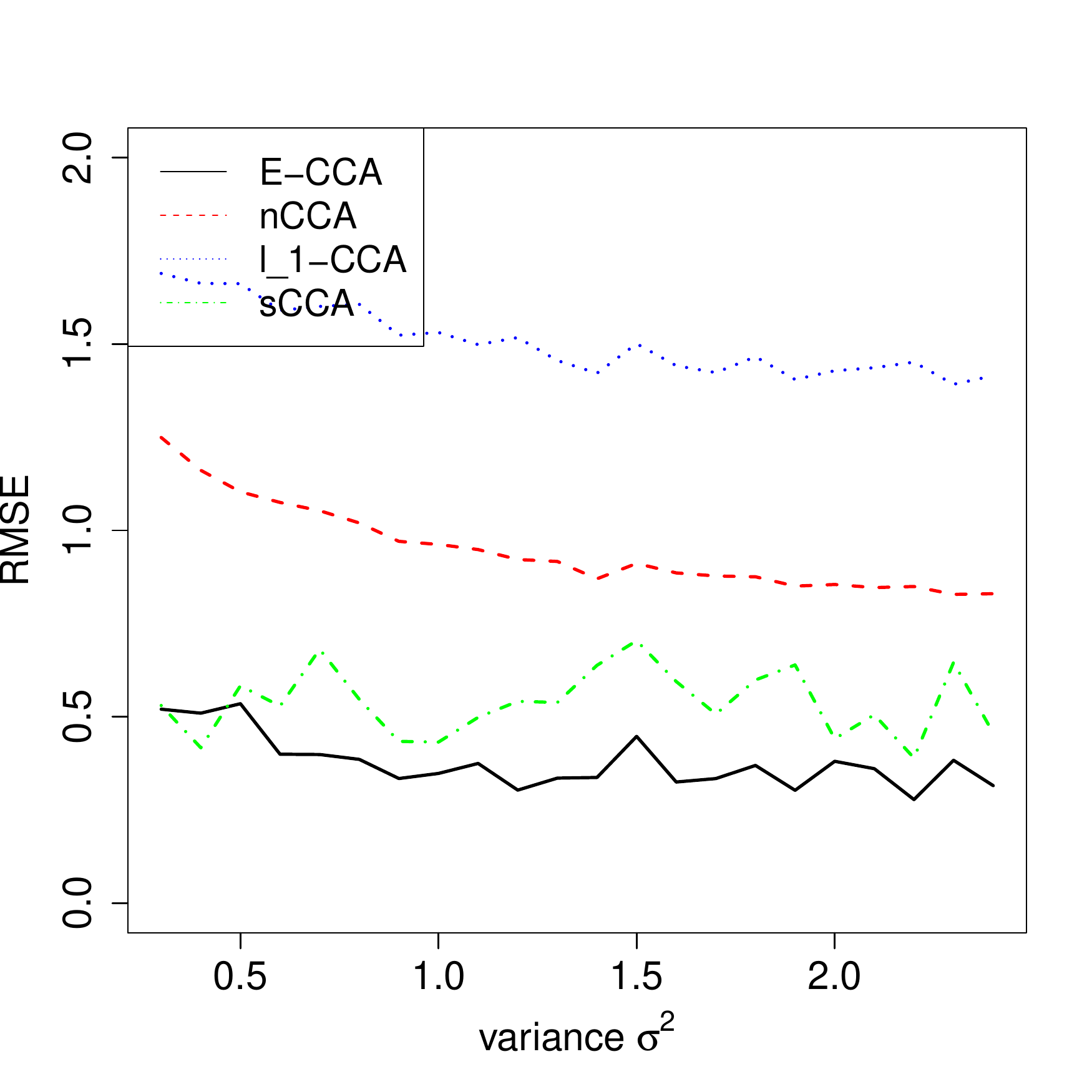}
        \caption{$(\hat{{\rm E}}\sum_{i=1}^3\|\hat a_i - a_i\|_2^2)^{1/2}$.}
    \end{subfigure}
    \begin{subfigure}[b]{0.24\textwidth}
        \centering
        \includegraphics[width=\textwidth]{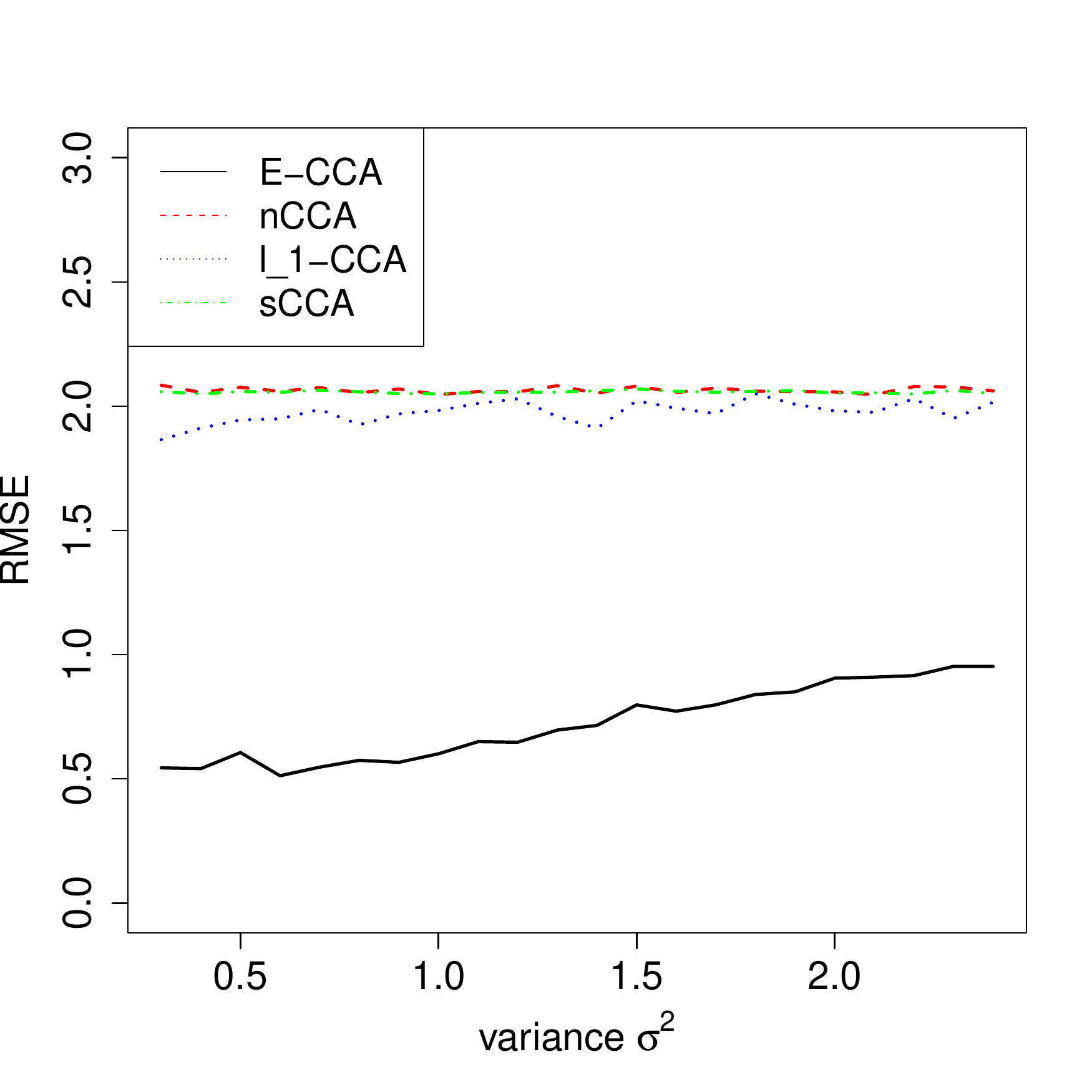}
        \caption{$(\hat{{\rm E}}\sum_{i=1}^3\|\hat b_i - b_i\|_2^2)^{1/2}$.}
    \end{subfigure}
        \begin{subfigure}[b]{0.24\textwidth}
        \centering
        \includegraphics[width=\textwidth]{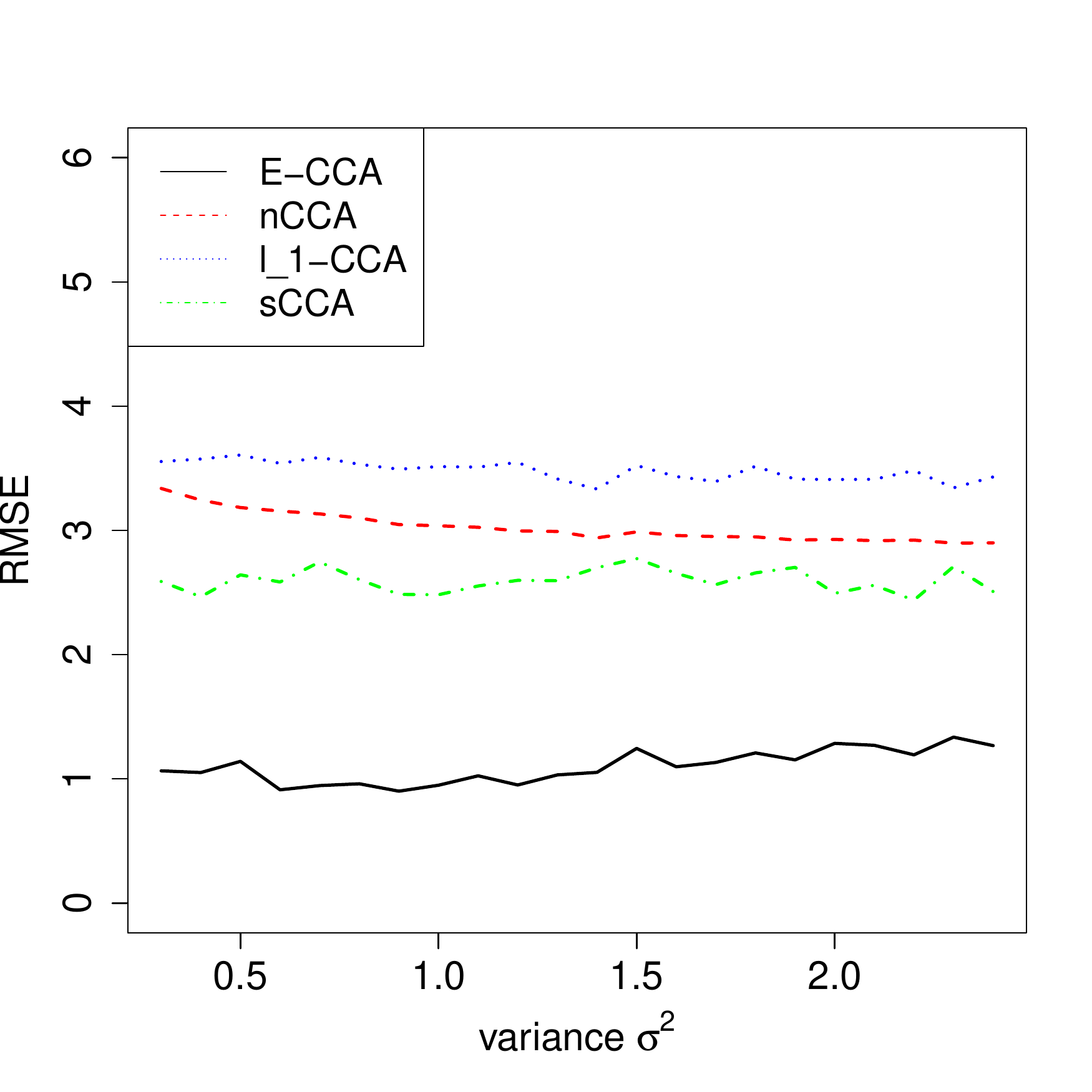}
        \caption{Total error.}
    \end{subfigure}
    \begin{subfigure}[b]{0.24\textwidth}
        \centering
        \includegraphics[width=\textwidth]{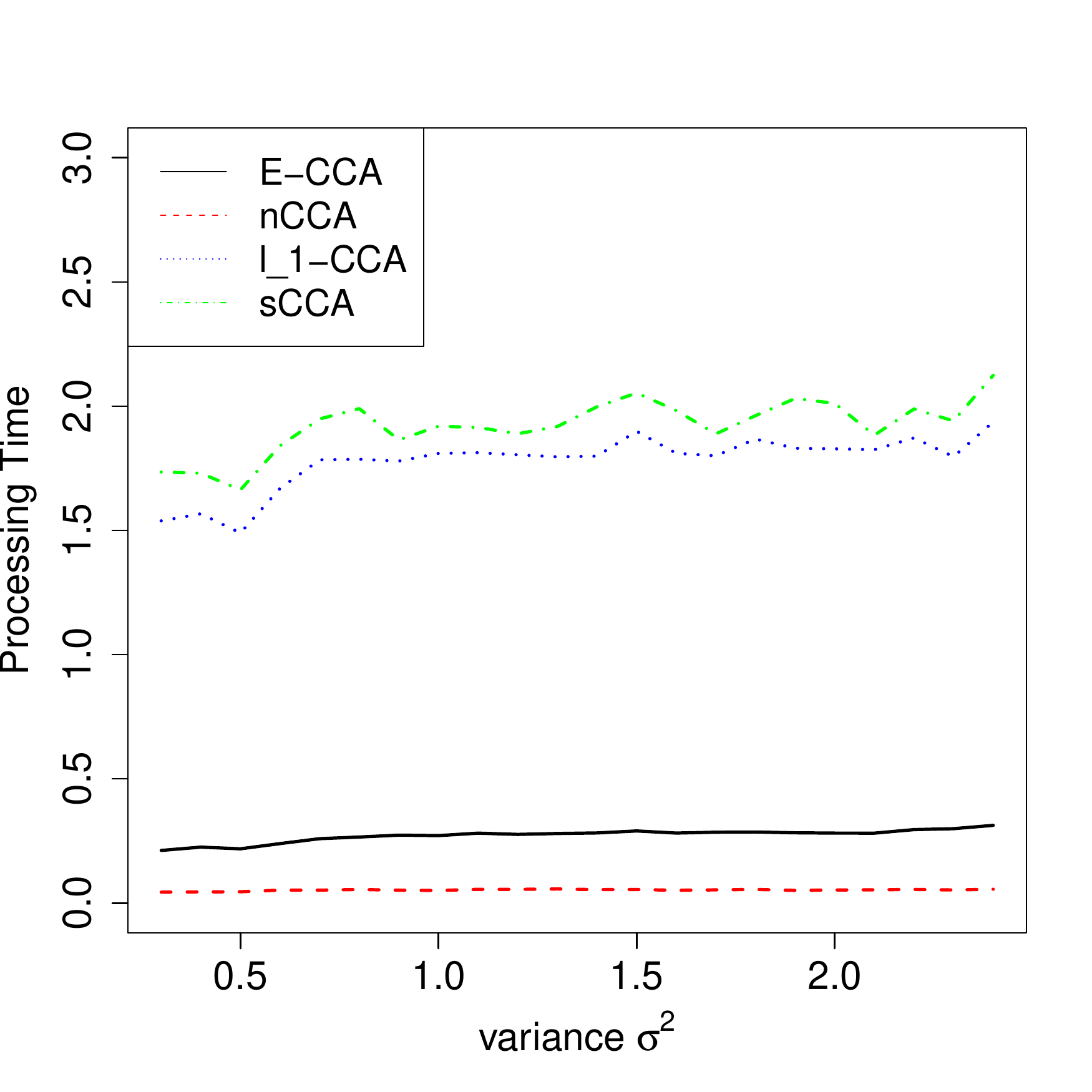}
        \caption{Processing time.}
    \end{subfigure}
    \begin{subfigure}[b]{0.24\textwidth}
        \centering
        \includegraphics[width=\textwidth]{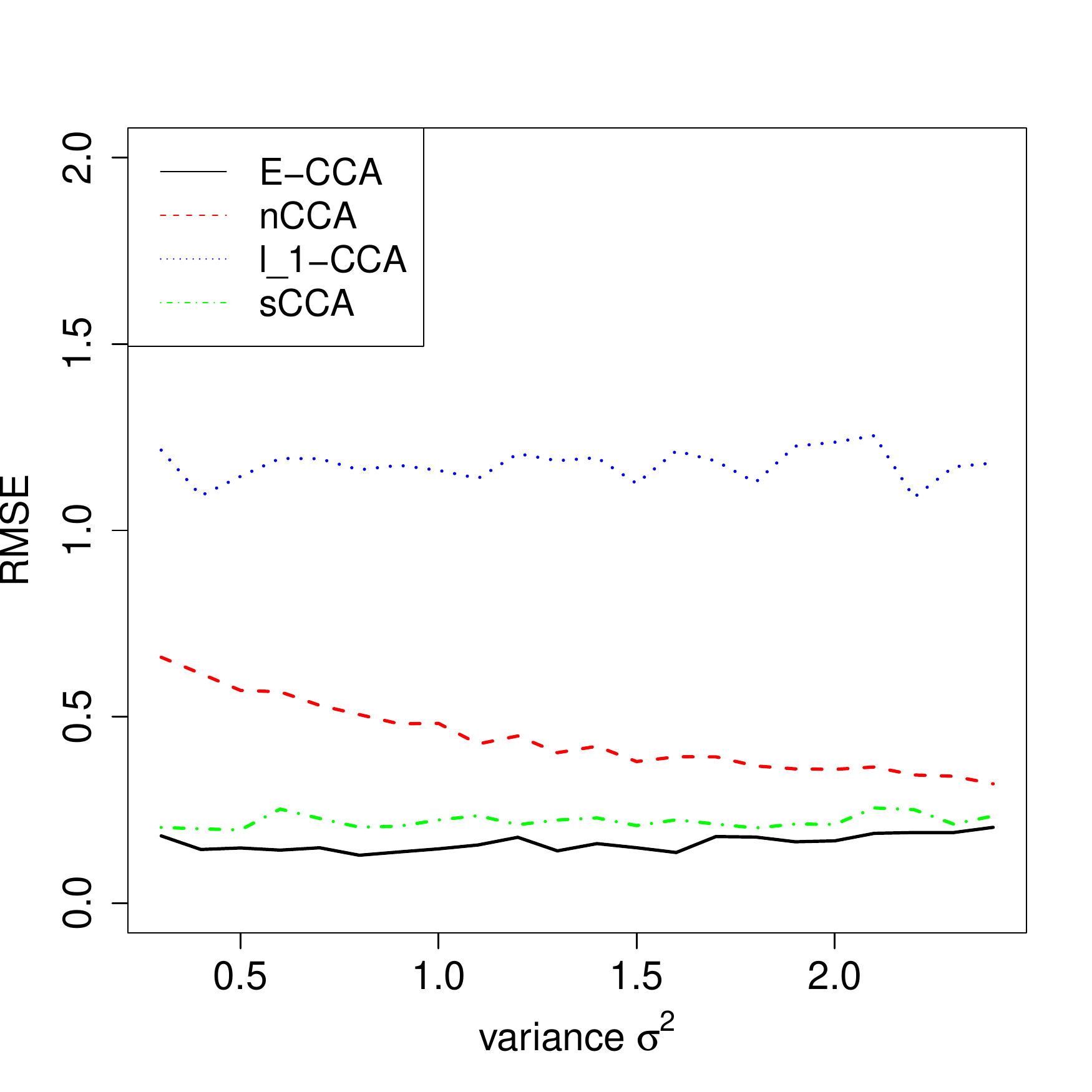}
        \caption{$(\hat{{\rm E}}\sum_{i=1}^3\|\hat a_i - a_i\|_2^2)^{1/2}$.}
    \end{subfigure}
        \begin{subfigure}[b]{0.24\textwidth}
        \centering
        \includegraphics[width=\textwidth]{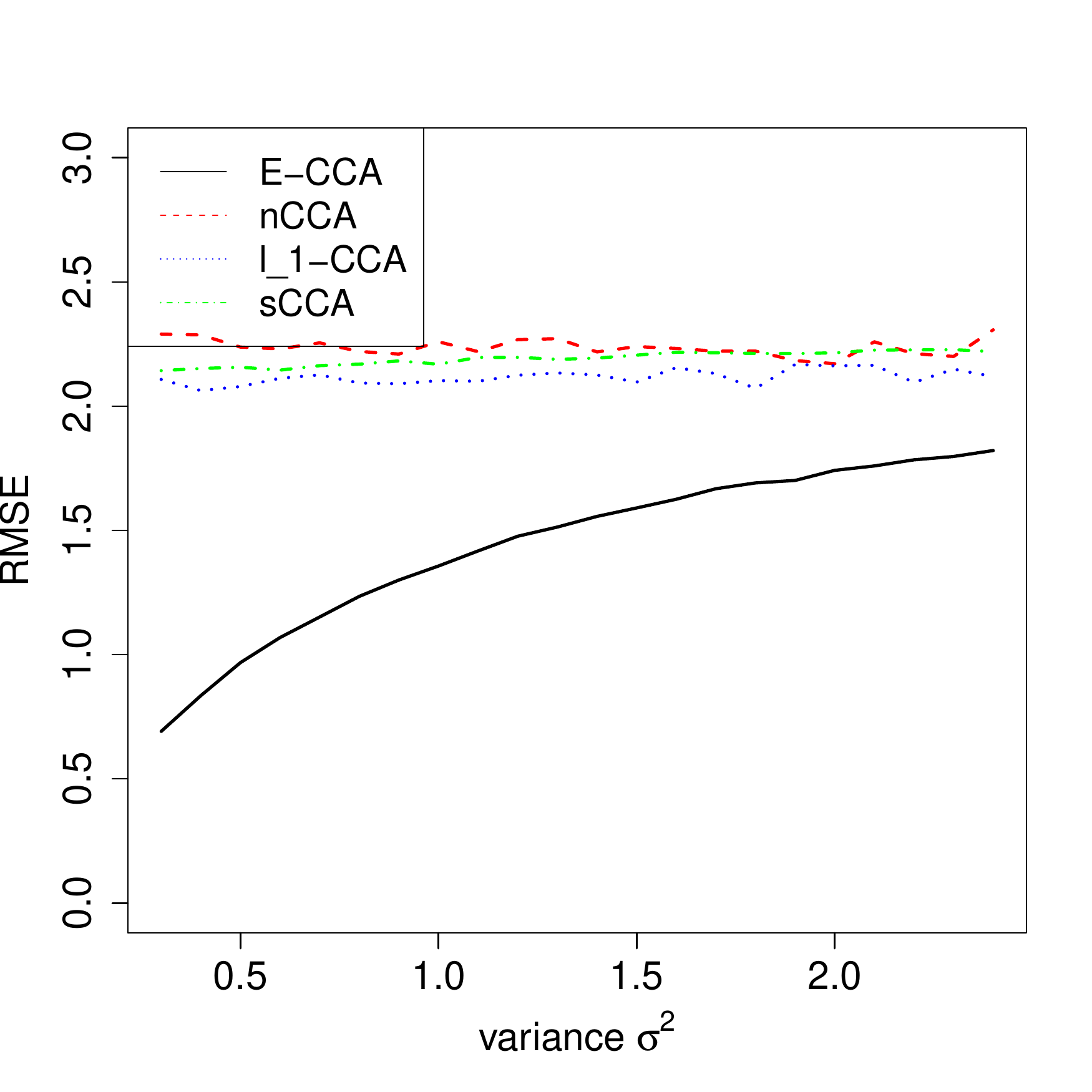}
        \caption{$(\hat{{\rm E}}\sum_{i=1}^3\|\hat b_i - b_i\|_2^2)^{1/2}$.}
    \end{subfigure}
    \begin{subfigure}[b]{0.24\textwidth}
        \centering
        \includegraphics[width=\textwidth]{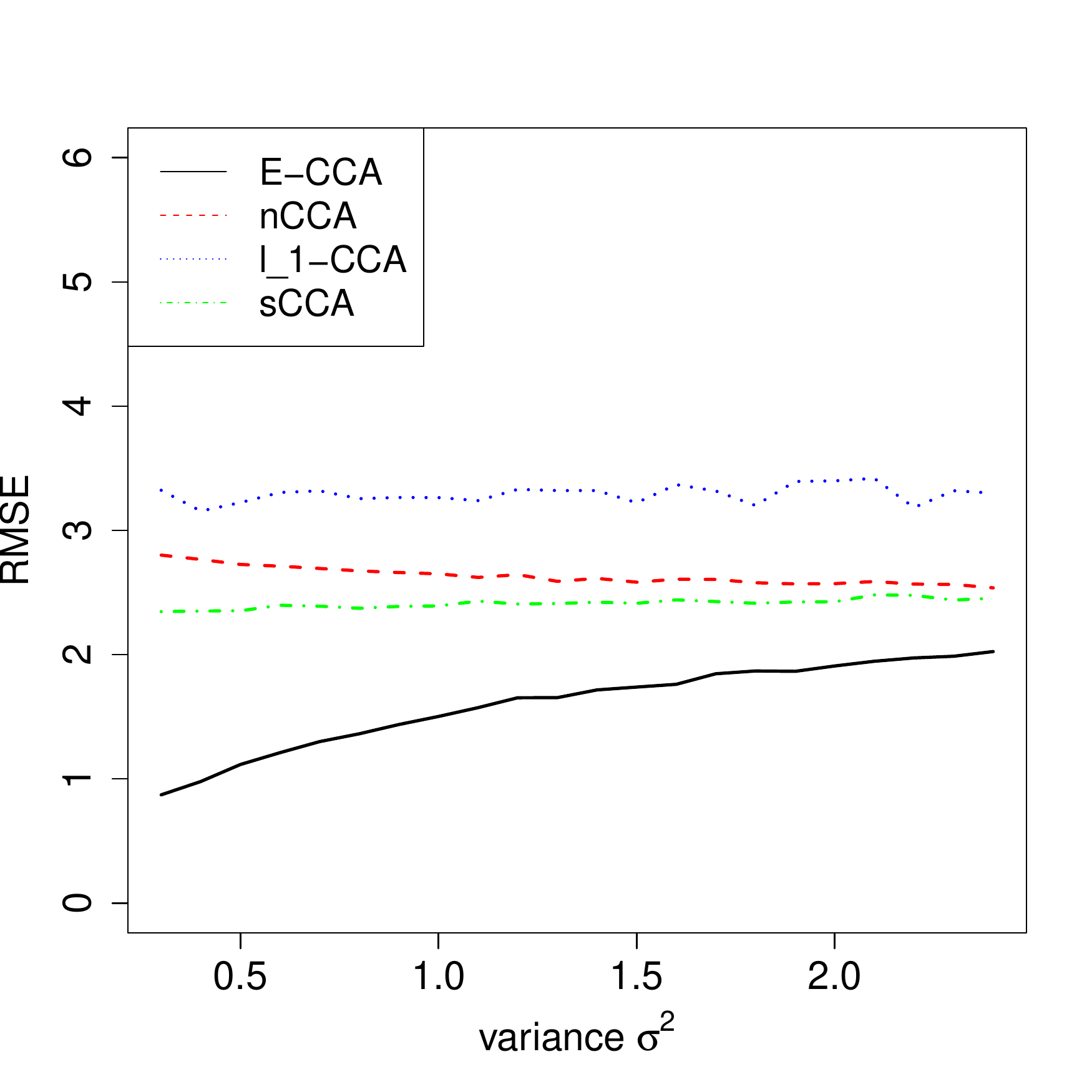}
        \caption{Total error.}
    \end{subfigure}
    \begin{subfigure}[b]{0.24\textwidth}
        \centering
        \includegraphics[width=\textwidth]{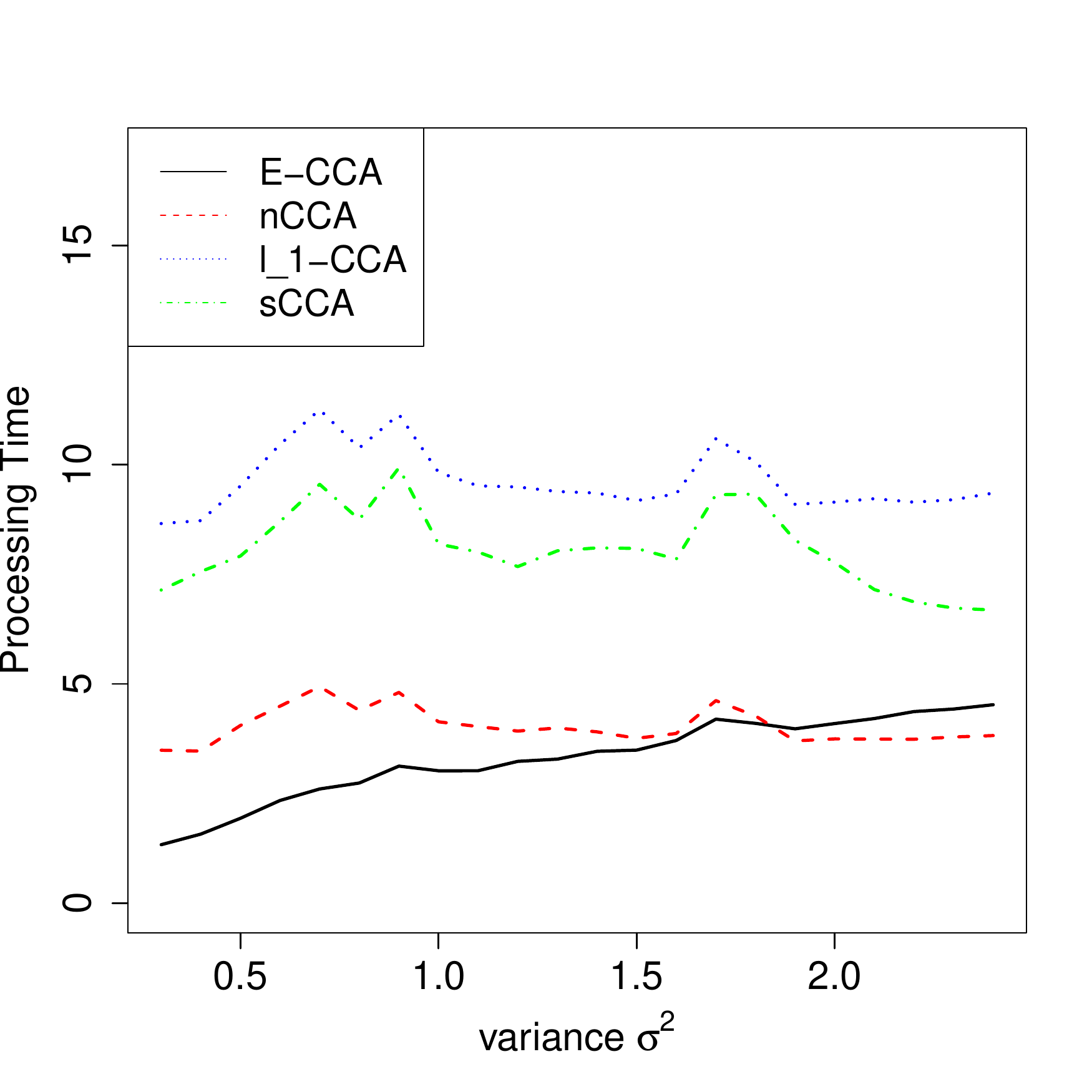}
        \caption{Processing time.}
    \end{subfigure}
    \caption{The average of approximated root mean squared prediction errors and processing time for Case 2. The processing time is in seconds.  In subfigures (c) and (g), Total error = $(\hat{{\rm E}}\sum_{i=1}^3\|\hat a_i - a_i\|_2^2)^{1/2} +(\hat{{\rm E}}\sum_{i=1}^3\|\hat b_i - b_i\|_2^2)^{1/2}$. \textbf{Row 1:} $(p,d)=(100,5)$. \textbf{Row 2:} $(p,d)=(500,3)$. 
    }
    \label{fig:case2ab}
\end{figure}

It can be seen that sCCA performs well in most cases when estimating the canonical vectors of first variables $a_1,a_2$, and $a_3$.  $l_1$-CCA cannot provide a consistent estimator of the canonical vectors of first variables. nCCA does not perform well in Case 2. However, when we turn to look at the estimation of the canonical vectors of second variables $b_1,b_2$, and $b_3$, we can see that nCCA, $l_1$-CCA and sCCA cannot provide a consistent estimator in most cases. This indicates that these methods are not appropriate when the dimensions of $x$ and $y$ are quite different, because these methods do not utilize the low dimensional structure of $y$. E-CCA works well on the estimation of $b_1,b_2$, and $b_3$, and has the smallest total prediction error among all methods in most cases. We can also see the total prediction error of E-CCA increases as $\sigma^2$ increases, which is natural because the accuracy of the estimation of coefficients using \eqref{regumodel} is influenced by the variance $\sigma^2$. As for the computation time, it can be seen that only the naive nCCA can be faster than our algorithm when $p$ is also relatively small. Nevertheless, as we have explained before, nCCA does not provide a sparse canonical vector, which is not desired in many cases. In addition, nCCA becomes less efficient when $p$ is large. $l_1$-CCA and sCCA have much more computation time than E-CCA. For example, E-CCA only takes about $7\%$ to $15\%$ time of $l_1$-CCA and sCCA when $p=100$. E-CCA is more efficient if the difference between $p$ and $d$ is large. In Figure 1 (k)(l), it can be seen that our method achieves a smaller total error than nCCA and $l_1$-CCA, and a larger error than sCCA. However, in order to achieve this smaller error, sCCA takes 10 times as much computation time as E-CCA. For one iteration (that is, one replicate), nCCA takes about 370 seconds, $l_1$-CCA takes about 85 seconds, sCCA takes about 210 seconds, while E-CCA only takes around 20 seconds. It takes about 1.5 hour for sCCA to finish 25 replicates in the simulation, while E-CCA only takes 8 minutes. This implies that E-CCA is much more efficient than $l_1$-CCA and sCCA. Note that here we do not apply parallel computing to E-CCA in these numeric examples. If parallel computing is available, E-CCA can be more efficient.

\section{Real data examples}\label{seccasestudy}

\subsection{Analysis of GTEx thyroid histology images}\label{seccasestudy1}

The GTEx project offers an opportunity to explore the relationship between imaging and gene expression, while also considering the effect of a clinically-relevant trait. We obtained the original GTEx thyroid histology images (see Figure \ref{fig:realdata2} for example) from the Biospecimen Research Database (\url{https://brd.nci.nih.gov/brd/image-search/searchhome}). These image files are in Aperio SVS format, a single-file pyramidal tiled TIFF. The RBioFormats R package (\url{https://git- hub.com/aoles/RBioFormats}), which interfaces the OME Bio-Formats Java library (\url{https://www.openmicroscopy.org/bio-formats}), was used to convert the files to JPEG format \cite{paulzhou2020}.  These images were further processed using the Bioconductor package EBImage \cite{pau2010ebimage}. Following the method proposed by \cite{barry2018histopathological} to segment individual tissue pieces, the average intensity across color channels was calculated, and adaptive thresholding was performed to distinguish tissue from background. A total of 108 independent Haralick image features were extracted from each tissue piece by calculating 13 base Haralick features for each of the three RGB color channels and across three Haralick scales by sampling every 1, 10, or 100 pixels \cite{paulzhou2020}. The features were log2-transformed and normalized to ensure feature comparability across samples. 
\begin{figure}[!ht]
    \centering
    \includegraphics[width=6.5in]{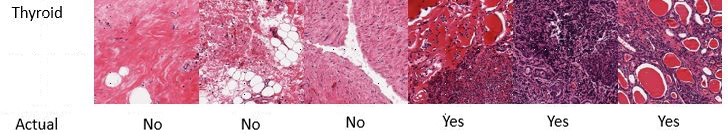}
    \caption{Examples of Hashimoto's thyroiditis negative/positive GTEx samples.}
    \label{fig:realdata2}
\end{figure}

To obtain a trait with clinical relevance, we also downloaded the thyroiditis Hashimoto pathology data from the GTEx Portal (\url{https://www.gtexportal.org/home/histologyPage}). Sex and age are also provided. The phenotype Hashimoto's thyroiditis was the presence (coded 1) or absence (0) of a particular pathology.

For thyroid, a subset of these subjects (570) also had gene expression data from RNA-Seq. The v8 release is available on the GTEx Portal (\url{https://www.gtexportal.org/home/datasets}). Gene read counts were normalized between samples using TMM, genes were selected based on expression thresholds explained in their paper \cite{aguet2019gtex}. In this example, we collect both processed image feature matrix $(Y)$ and gene expression data $(X)$ on these overlapped 570 subjects, with 37 cases of Hashimoto's thyroiditis. 

We applied eigenvector-based sparse CCA, other three methods in Section \ref{secnumeric}, and rgCCA to study the correlation between the processed image feature matrix and gene expression data. Since eigenvector-based sparse CCA works well for the low dimensional data $Y$, we first applied principal component analysis (PCA) to reduce the dimension of $Y$. We used the first half of principal components, denoted by $U=(u_1,...,u_d)$, and performed eigenvector-based sparse canonical correlation analysis on the transformed variables $UY$ and $X$. We randomly split the data into a training set $(Y^{train},X^{train})$ and test set $(Y^{test},X^{test})$ with ratio $5:1$ for 500 times. For each run, we obtained the first pair of estimated canonical vectors $\hat a_1$ and $\hat b_1$ by the methods mentioned in Section \ref{secnumeric}, and compared the correlations on the test set ${\rm Corr}((Y^{test})^\top \hat a_1,(X^{test})^\top \hat b_1)$. We found that nCCA is very sensitive to the value of the nugget parameter, so we did not include it in the comparison. For E-CCA, we further randomly split the the training set $(Y^{train},X^{train})$ to new training set $(Y^{newtrain},X^{newtrain})$ and validation set $(Y^{val},X^{val})$ with ratio 44:5, and select the parameter $\lambda_1$ that maximizes the canonical correlation on the validation set. The results obtained by E-CCA, $l_1$-CCA, sCCA and rgCCA are shown in Figure \ref{fig:real2results}.

\begin{figure}[!ht]
    \centering
    \includegraphics[width=0.5\textwidth]{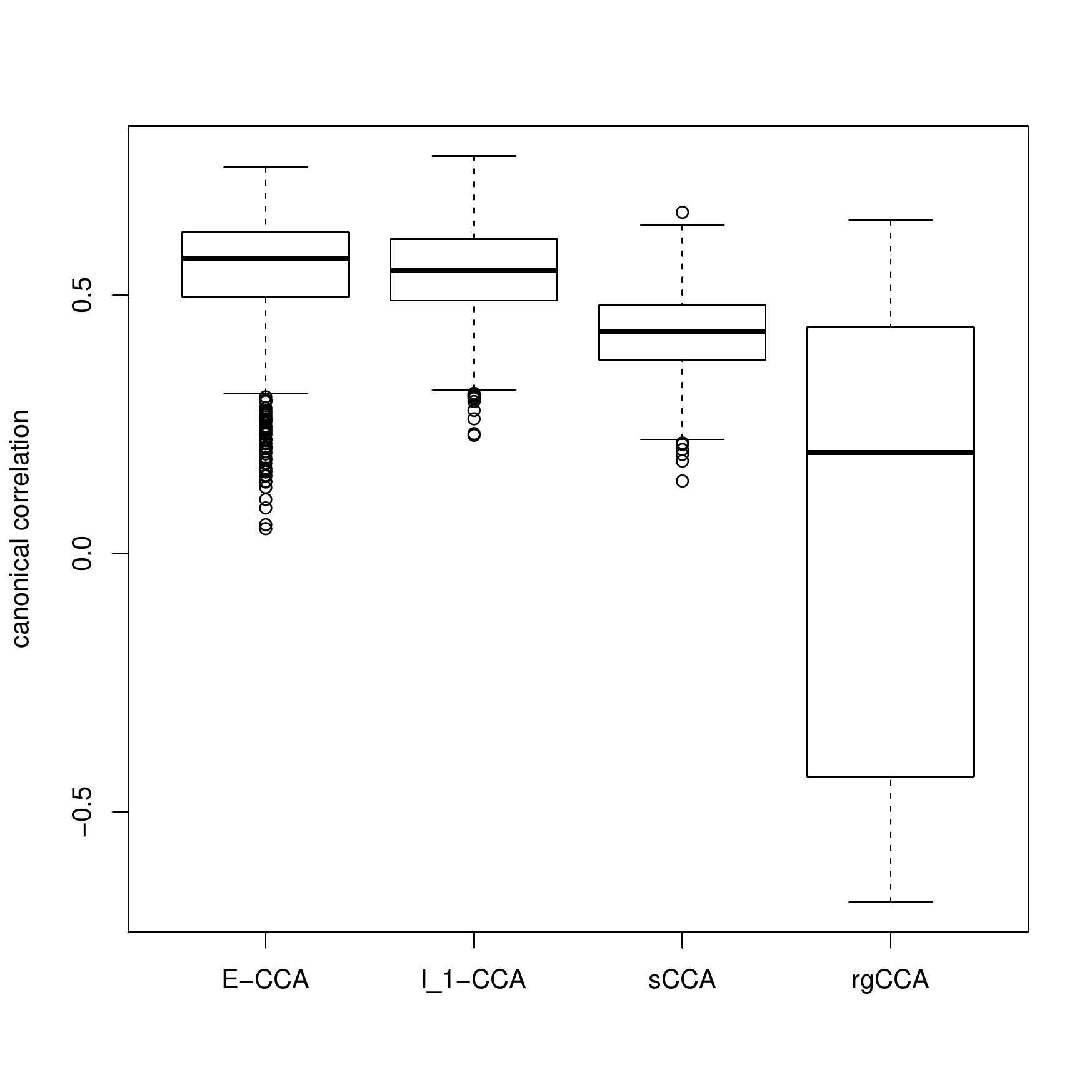}
    \caption{Boxplots of the GTEx thyroid cross-validated canonical correlations of processed image feature matrix and gene expression data.}
    \label{fig:real2results}
\end{figure}

From Figure \ref{fig:real2results}, we can see that rgCCA does not provide reliable estimation of the canonical variables in this study, and sCCA provides a smaller correlation between the estimated canonical variables on the test data. E-CCA is slightly better than $l_1$-CCA. In some cases, E-CCA provides relatively small correlations between the estimated canonical variables on the test data. This may be because we fixed the number of principal components. Therefore, a further study on adaptively choosing number of principal components is needed.

In order to explore the effect of a clinical phenotype on E-CCA, we performed E-CCA separately on the set of individuals without Hashimoto's thyroiditis (median E-CCA of 0.578, nearly the same as for the full dataset), and for individuals with Hashimoto's thyroiditis (median E-CCA of 0.375). The dramatic change in estimated correlation by case/control status provides a window into potential additional uses of sparse CCA methods, e.g. by using the contrast in canonical correlation by phenotype to improve omics-based phenotype prediction.

\subsection{Analysis of SNP genotype data and RNA-seq gene expression data}
We tested eigenvector-based sparse CCA, other three methods in Section \ref{secnumeric}, and rgCCA using data from the GTEx V8 release (\url{https://www.gtexportal.org/home/}), including genotype data from a selected set of SNPs and RNA-seq gene expression data from $n=208$ liver tissue samples. SNPs were coded from $0$-$2$ as the number of minor alleles, and RNA-seq expression data were normalized using simple scaling.  Among the problems that arise in such datasets is the powerful mapping of sets of SNPs that are collectively associated with expression traits. Here we use CCA to demonstrate a proof of principle for finding such collective association in a biological pathway. We selected SNP sets for each gene by grouping SNPs located within 5kb of a gene’s transcription start site (TSS).  Then we grouped genes into gene sets based on the canonical pathways listed in the Molecular Signatures Database (MSigDB) v7.0 (\url{https://www.gsea-msigdb.org/gsea/msigdb/index.jsp}).  These gene sets are canonical representations of a biological process compiled by domain experts.  We analyzed 2,072 pathways with a size between 5-200 genes. So for each pathway, we have a genotype matrix $X$ with $p$ SNPs and $n$ samples and expression matrix $Y$ with $d$ genes and $n$ samples, with $p>n>d$.  For each method, a permutation-based $p$-value was calculated after performing 1,000 permutations.

Here we focus on two pathways of potential biological relevance in the liver, with strong eQTL evidence.  One is the keratinization pathway (\url{https://www.reactome.org/content/detail/R-HSA-6805567}), which included 72 genes and 3,005 SNPs from our dataset. The $p$-values were $0.006, 0.10, 0.68, 0.14$, and $0.81$ for E-CCA, nCCA, sCCA, $l_1$-CCA, and rgCCA respectively.  Three genes,  KRT13, KRT4, and KRT5, showed values of $|\hat{a}|$ that are much larger than those of the remaining genes, while the values of $\hat{b}$ are spread more uniformly across the SNPs. Keratins are important for the mechanical stability and integrity of epithelial cells and liver tissues.  They play a role in protecting liver cells from apoptosis, against stress, and from injury, and defects may predispose to liver diseases \cite{moll2008human}.  
Figure \ref{fig:realdata3} shows heatmap plots and Manhattan-style line plots showing the absolute values of $\hat{a}$ and $\hat{b}$, in which SNPs (rows) and genes (columns) are ordered by genomic position. 

\begin{figure}[!h]
    \centering
    \includegraphics[width=0.7\textwidth, height = 3.5in]{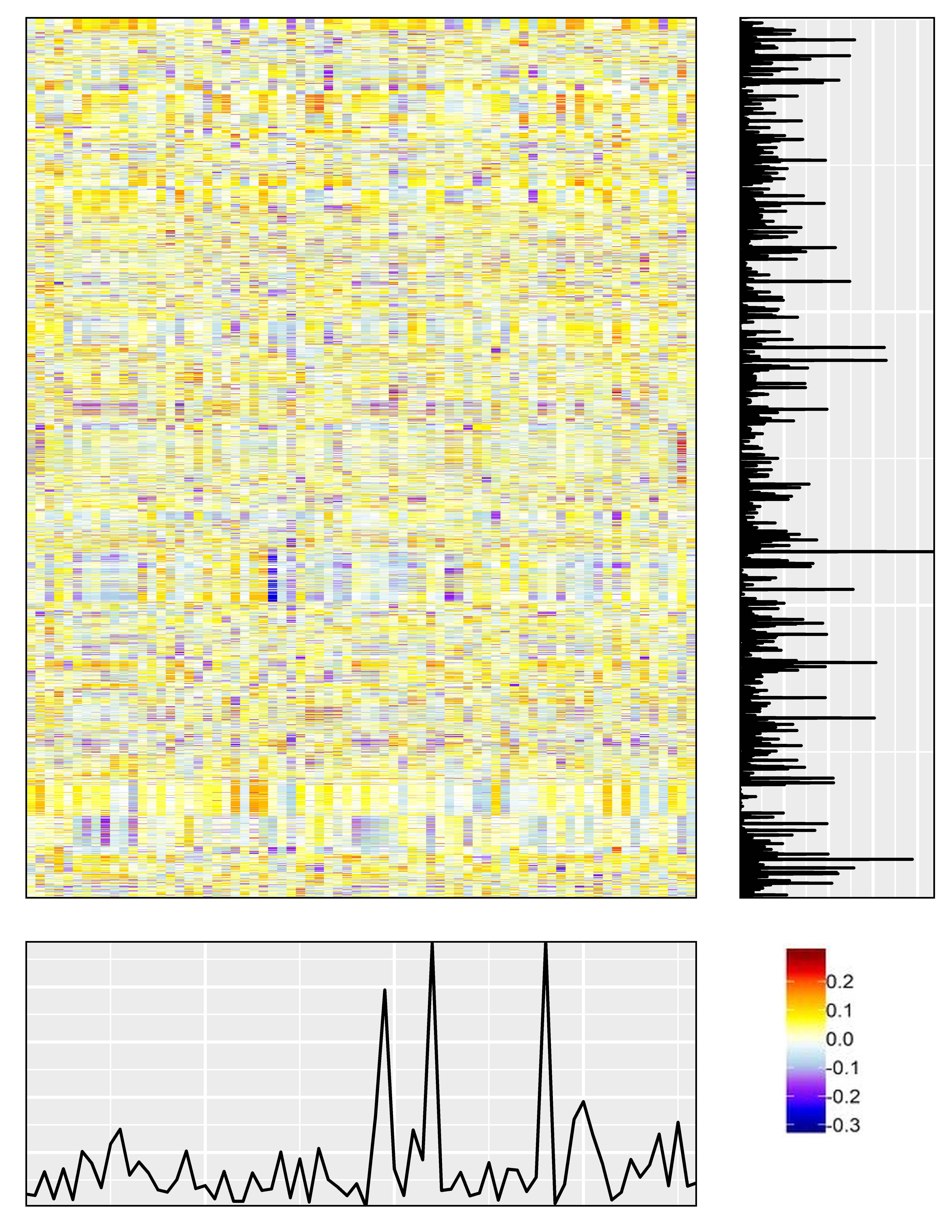}
    \caption{The heatmap plots and Manhattan-style line plots showing the absolute values of $\hat{a}$ and $\hat{b}$, in which SNPs (rows) and genes (columns) are ordered by genomic position.}
    \label{fig:realdata3}
\end{figure}

A smaller pathway is the synthesis of ketone bodies (\url{https://www.reactome.org/content/detail/R-HSA-77111}), with 8 genes and 265 SNPs.  The $p$-values were 0.003, 0.16, 0.52, 0.35, and 0.66 for E-CCA, nCCA, sCCA, $l_1$-CCA, and rgCCA respectively.  Again three genes, ACSS3, BDH2, and BDH1 showed $|\hat{a}|$ of greater magnitude than the others. Ketone bodies are metabolites derived from fatty and amino acids and are mainly produced in the liver.  Both in the biosynthesis of ketone bodies (ketogenesis) and in ketone body utilization (ketolysis), inborn errors of metabolism are known, resulting in various metabolic diseases \cite{sass2012inborn}. 
\subsection{Analysis of human gut microbiome data }

We applied eigenvector-based sparse CCA to a microbiome study conducted at University of Pennsylvania \cite{chen2012structure}. The study profiled 16S rRNA in the human gut and measured components of nutrient intake using a food frequency questionnaire for 99 healthy people. Microbiome OTUs were consolidated at the genus level, with $d=40$ relatively common genera considered (i.e., $Y$ was a $40\times 99$ OTU abundance matrix). Following \cite{chen2012structure}, the daily intake for $p=214$ nutrients was calculated for each person, and regressed upon energy consumption, and the residuals used as a processed nutrient intake $214\times 99$ matrix $X$.

ssCCA \cite{chen2012structure} identified 24 nutrients and 14 genera whose linear combinations gave a cross-validated canonical correlation of 0.42 between gut bacterial abundance and nutrients. Eigenvector-based sparse CCA reached a canonical correlation of 0.60. To test the canonical correlation between gut bacterial abundance and nutrients, we permuted columns of the nutrient matrix 1,000 times, and calculated the canonical correlation between them using the four CCA methods described in Section \ref{secnumeric} and rgCCA \cite{tenenhaus2011regularized,rgccapackage}. These correlations constitute a null distribution for each method, to which we compared the respective observed canonical correlation. The E-CCA method was significant at the 0.05 level, with $p$-value 0.025. Of the remaining methods, only sCCA and nCCA (with a large nugget parameter) also provided significant $p$-values. However, results from nCCA appeared highly sensitive to the nugget parameter, and range of choices for nugget parameters produced nonsignificant $p$-values. $l_1$-CCA and rgCCA did not appear to provide insightful results for this dataset. The heatmap of the covariance matrix ($XY^\top $) is shown in Figure \ref{fig:realdata}. The marginal plots of absolute values of $\hat a_1$ and $\hat b_1$ provide insights for the relative weighting of OTUs and nutritional components toward the overall canonical correlation, i.e. larger values correspond to greater weight for that component.  

\begin{figure}[!ht]
    \centering
    \includegraphics[width=0.7\textwidth]{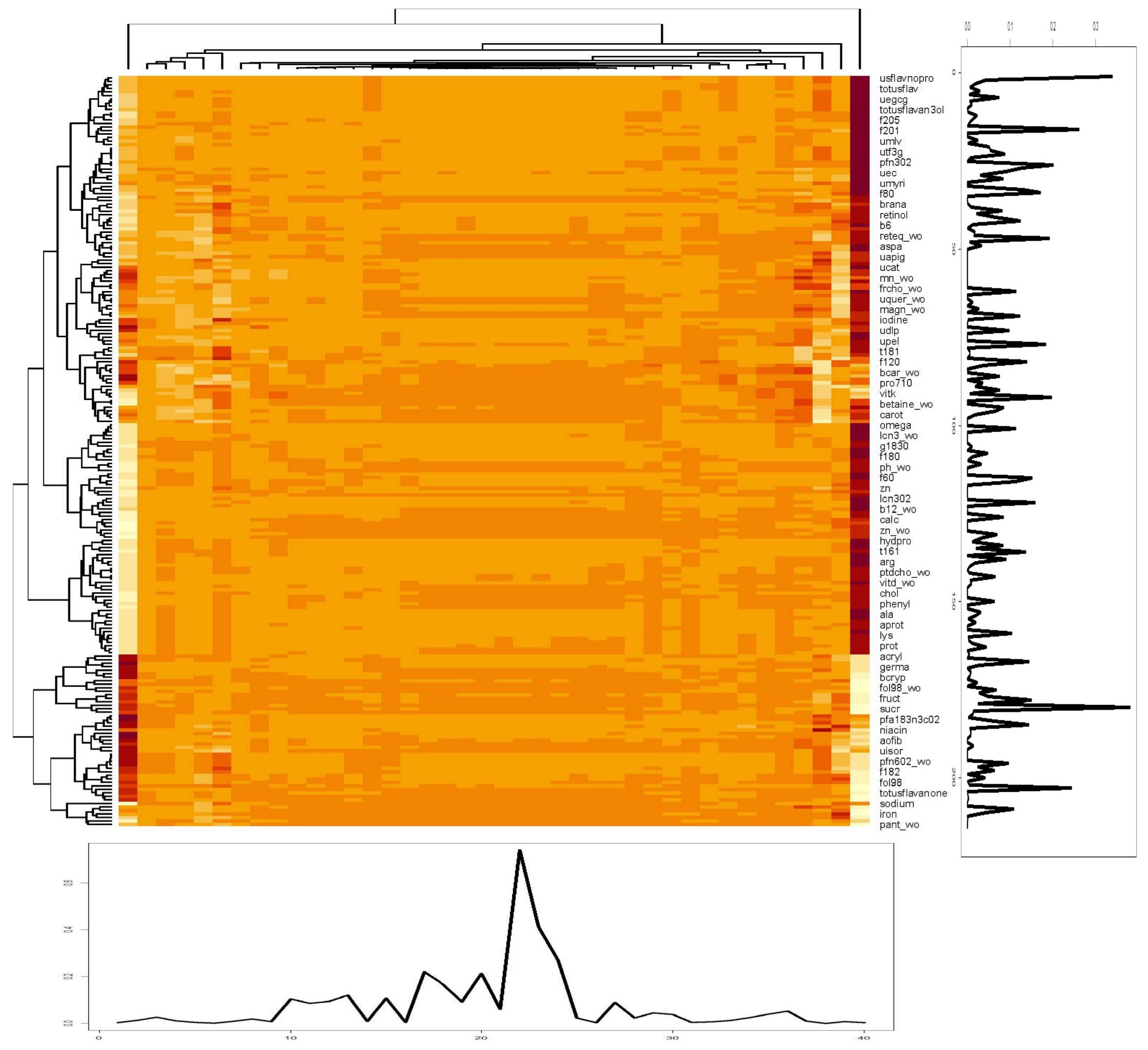}
    \caption{The heatmap of the covariance matrix between gut bacterial abundance and nutrients.}
    \label{fig:realdata}
\end{figure}

\section{Conclusions and discussion}\label{secconclu}
In this work, we proposed eigenvector-based sparse canonical correlation analysis, which can be applied to data where the dimensions of two variables are very different. Our method can provide $K$ pairs of canonical vectors simultaneously for $K>1$, and can be implemented by an efficient algorithm based on Lasso. The implementation is straightforward. The computation time is small, and can be further significantly decreased if parallel computing is available. We show the consistency of the estimated canonical pairs in the case that the dimension of one variable can increase as an exponential rate in comparison to the sample size. The dimension of the other variable should be smaller than the sample size. We also present numerical studies to show the efficiency of our algorithm and real data analysis to validate our methodology. 

As pointed by a reviewer, CCA and partial least squares regression (PLS) are very relevant. Both methods are used to find relationship between two random variables, while CCA maximizes the correlation and PLS maximizes the covariance. The relationship and comparison between CCA and PLS have been studied in several literature \cite{sun2009equivalence,grellmann2015comparison,cserhati1998comparison}, while the relationship between sparse CCA and PLS is not clear. Furthermore, based on the similarity of CCA and PLS, we believe that our methods can be also generalized or combined with other PLS methods, like sparse Multi-Block Partial Least Squares \cite{li2012identifying}. These generalization and combination will be studied in the future. 

We consider the \textit{unbalanced} case, where the dimensions of two variables are much different. In practice, there are also some cases that are \textit{balanced}, i.e., the dimensions of two variables are comparable but both much larger than the sample size. One straightforward potential extension is to apply two sets of Lasso problems. Specifically, we might set $Y$ as the dependent variable and $X$ as the independent variable in the first set of Lasso problems, and set $X$ as the dependent variable and $Y$ as the independent variable in the second set of Lasso problems. However, the number of Lasso optimizations is very large, which leads to the inefficiency of the algorithm. One possible remedy is to apply principal components analysis to reduce the dimension of one variable. This approach has shown its potential in the real data analysis. However, the theoretical justification is currently lacking. Eigenvector-based sparse canonical correlation analysis also shows great potential for prediction problems, since it can provides $K$ pairs of canonical vectors efficiently and simultaneously for $K>1$. These possible extensions of eigenvector-based sparse canonical correlation analysis to the balanced case will be pursued in the future work.

\section{Proof of Theorem \ref{thmonelasso}}\label{apppfonelasso}
Let $A=(a_{ij})_{ij}\in \RR^{m\times n}$ and $\|A\|_p$ be the $p$-norm of a matrix $A$. With an abuse of notation, we use $\|\cdot\|_p$ for both $p$-norm of a matrix and $l_p$ norm of a vector. Let $\|A\|_{\max} = \max |a_{ij}|$. In the special cases $p=1,2,\infty$,
\begin{align*}
    \|A\|_1  & = \max_{1\leq j \leq n}\sum_{i=1}^m|a_{ij}|,\quad \|A\|_2 = \sigma_{\max}(A),\quad
    \|A\|_\infty  = \max_{1\leq i \leq m}\sum_{j=1}^n|a_{ij}|,
\end{align*}
where $\sigma_{\max}(A)$ is the largest singular value of matrix $A$. These matrix norms are equivalent, which are implied by the following inequality,
\begin{align*}
    \frac{1}{\sqrt{n}}\|A\|_\infty \leq \|A\|_2 \leq \sqrt{m}\|A\|_\infty,
    \frac{1}{\sqrt{m}}\|A\|_1  \leq \|A\|_2 \leq \sqrt{n}\|A\|_1,
    \|A\|_{\max} \leq \|A\|_2 \leq \sqrt{mn}\|A\|_{\max}.
\end{align*}

We first present some lemmas used in this proof. Lemma \ref{lemb3ning} states the consistency of $\hat \beta_k$ obtained by \eqref{sapregumodel}. Lemma \ref{lembernstein} is the Bernstein inequality. Lemma \ref{lemconcenEqq} is the concentration inequality for sub-Gaussian random vectors. Lemma \ref{lemmaLinearSystems} describes the accuracy of solving linear systems; see Theorem 2.7.3 in \cite{van1983matrix}. Lemma \ref{lemperteigenvec} states the eigenvector sensitivity for a pertubation of a matrix, which is a slight recasting of Theorem 4.11 in \cite{stewart1973error}; also see Corollary 7.2.6 in \cite{van1983matrix}.

\begin{lemma}\label{lemb3ning}
Suppose the conditions of Theorem 2.1 hold. Then with probability at least $1-C_1 p^{-1}d$,
\begin{align*}
    \max_{1\leq k \leq d}\|\hat \beta_k - \beta_k^*\|_1\leq C_2s^*\sqrt{\frac{\log p }{n}}, \max_{1\leq k \leq d}\|\hat \beta_k - \beta_k^*\|_2\leq C_3\sqrt{\frac{s^*\log p }{n}}.
\end{align*}
In addition, $\max_{1\leq k \leq d}(\hat \beta_k - \beta_k^*)^\top H_X(\hat \beta_k - \beta_k^*)\lesssim s^*\log p/n$, where $H_X= n^{-1}\sum_{i=1}^n X_iX_i^\top $.
\end{lemma}
\begin{remark}
In \cite{ning2014general}, the model is 
\begin{align*}
  \min_{\beta} \frac{1}{2n}\sum_{i=1}^n  (y_{i} - \beta^\top X_i)^2 + \lambda_1 \|\beta\|_1
\end{align*}
instead of \eqref{sapregumodel}. This is the reason that $\lambda_1\asymp \sqrt{n\log p}$ in Theorem \ref{thmonelasso}, not $\sqrt{\log p/n}$ as in \cite{ning2014general}.
\end{remark}
\begin{proof}[Proof of Lemma \ref{lemb3ning}]
By Lemma B.3 of \cite{ning2014general}, we have for a fixed $k$, with probability at least $1-C_1 p^{-1}$, 
\begin{align*}
    \|\hat \beta_k - \beta_k^*\|_1\leq C_2s^*\sqrt{\frac{\log p }{n}},
    \|\hat \beta_k - \beta_k^*\|_2\leq C_3\sqrt{\frac{s^*\log p }{n}},
    (\hat \beta_k - \beta_k^*)^\top H_X(\hat \beta_k - \beta_k^*)\lesssim \frac{s^*\log p}{n}.
\end{align*}
Then the results follow the union bound inequality.
\end{proof}

\begin{lemma}\label{lembernstein}
Let $X_i$'s be independent mean zero sub-Gaussian variables. There exists a constant $C>0$ such that for any $t>0$, 
\begin{align*}
    {\rm Pr}\left( \frac{1}{n} \left|\sum_{i=1}^n X_i \right|\geq t\right) \leq 2\exp(-Cnt^2).
\end{align*}
\end{lemma}

\begin{lemma}\label{lemconcenEqq}
Let $Q_i\in \mathbb{R}^d$ be sub-Gaussian random vectors for $i\in\{1,...,n\}$. We have 
\begin{align*}
    {\rm Pr}(\|H_Q - {\rm E}(QQ^\top )\|_{\max}\geq t)\leq 2q^2\exp(-C_1 nt^2),
\end{align*}
for some constants $C_1,C_2>0$, where $H_Q= n^{-1}\sum_{i=1}^n Q_iQ_i^\top $. 
\end{lemma}
\begin{proof}
The results follow the union bound inequality and Lemma \ref{lembernstein}.
\end{proof}

\begin{lemma}\label{lemmaLinearSystems}
Let $A,\tilde{A}\in \RR^{d\times d}$, and $b,\tilde{b}\in \RR^d$. Suppose ${Ax}={b}$ and $\tilde{{A}}\tilde{{x}}=\tilde{{b}}$ with $\|\tilde{{A}}-{A}\|_2\leq \delta \|{A}\|_2$, $\|\tilde{{b}}-{b}\|_2\leq \delta \|{b}\|_2$, and $\kappa({A})=r/\delta<1/\delta$ for some $\delta>0$. Then, $\tilde{{A}}$ is non-singular,
\begin{align*}
\frac{\|\tilde{{x}}\|_2}{\|{x}\|_2}\leq \frac{1+r}{1-r},
\frac{\|\tilde{{x}}-{x}\|_2}{\|{x}\|_2}\leq \frac{2\delta}{1-r}\kappa({A}),
\end{align*}
where $\kappa({A})=\|{A}\|_2\|{A}^{-1}\|_2$.
\end{lemma}

\begin{lemma}\label{lemperteigenvec}
Let $A,E \in \RR^{d\times d}$ and $Q = [q_1,Q_2] \in \RR^{d\times d}$ is orthogonal, where $q_1\in \RR^d$. Let
\begin{align*}
    Q^\top AQ = \left[
    \begin{array}{cc}
        \lambda &  v^\top \\
        0 & T_{22}
    \end{array}\right], 
    Q^\top EQ = \left[
    \begin{array}{cc}
        \epsilon &  r^\top \\
        \delta & E_{22}
    \end{array}\right].
\end{align*}
If $\sigma = \sigma_{\min}(T_{22} - \lambda I)>0$ and 
\begin{align*}
    \|E\|_2 \left(1 + \frac{5\|v\|_2}{\sigma}\right)\leq\frac{\sigma}{5},
\end{align*}
then there exists $u\in \RR^{d-1}$ with
\begin{align*}
    \|u\|_2\leq 4\frac{\|\delta\|_2}{\sigma}
\end{align*}
such that $\tilde{q}_1 = (q_1 +Q_2u)/\sqrt{1+u^\top u}$ is a unit 2-norm eigenvector for $A+E$.
\end{lemma}

Now we are ready to prove Theorem 1. We first show that $(YY^\top )^{-1}\hat BX(\hat BX)^\top $ is close to $\Sigma_{yy}^{-1}B_*\Sigma_{xx}B_*^\top $, then we apply Lemma \ref{lemperteigenvec} to show the consistency of canonical vectors. Without loss of generality, let $k=1$. If \eqref{thm1ineqrate} holds for $k=1$, then the results of Theorem 1 follow the union bound inequality.

By the triangle inequality, the $2$-norm of $(YY^\top )^{-1}\hat BX(\hat BX)^\top  - \Sigma_{yy}^{-1}B_*\Sigma_{xx}B_*^\top $ can be bounded by
\begin{align}\label{thmonelatotal}
    & \|(YY^\top )^{-1}\hat BX(\hat BX)^\top  - \Sigma_{yy}^{-1}B_*\Sigma_{xx}B_*^\top \|_2\nonumber\\
    \leq & \bigg\|(\frac{1}{n}YY^\top )^{-1}\frac{1}{n}\hat BX(\hat BX)^\top  - \Sigma_{yy}^{-1}\frac{1}{n}\hat BX(\hat BX)^\top  + \Sigma_{yy}^{-1}\frac{1}{n}\hat BX(\hat BX)^\top  - \Sigma_{yy}^{-1}\frac{1}{n}B_*X(B_*X)^\top  \nonumber\\
    & + \Sigma_{yy}^{-1}\frac{1}{n}B_*X(B_*X)^\top  - \Sigma_{yy}^{-1}B_*\Sigma_{xx}B_*^\top \bigg\|_2\nonumber\\
    \leq & \bigg\|(\frac{1}{n}YY^\top )^{-1}\frac{1}{n}\hat BX(\hat BX)^\top  - \Sigma_{yy}^{-1}\frac{1}{n}\hat BX(\hat BX)^\top \bigg\|_2 + \bigg\|\Sigma_{yy}^{-1}\frac{1}{n}\hat BX(\hat BX)^\top  - \Sigma_{yy}^{-1}\frac{1}{n}B_*X(B_*X)^\top \bigg\|_2 \nonumber\\
    & + \bigg\|\Sigma_{yy}^{-1}\frac{1}{n}B_*X(B_*X)^\top  - \Sigma_{yy}^{-1}B_*\Sigma_{xx}B_*^\top \bigg\|_2 \nonumber\\
    = & I_1 + I_2 + I_3.
\end{align}
We consider $I_2$ first. By Assumption \ref{onelassoassumcov}, we have
\begin{align}\label{pfthmonelassoi1}
    I_2 \leq & \bigg\|\Sigma_{yy}^{-1}\bigg\|_2 \bigg\|\frac{1}{n}\hat BX(\hat BX)^\top  - \frac{1}{n}B_*X(B_*X)^\top \bigg\|_2 \nonumber\\
    \leq & \frac{1}{K_1} \|B_*H_XB_*^\top  - \hat BH_XB_*^\top  + \hat BH_XB_*^\top  - \hat BH_X\hat B^\top \|_2\nonumber\\
    = & \frac{1}{K_1} \|(B_* - \hat B)H_X(B_* + \hat B)^\top \|_2\nonumber\\
    \leq & \frac{1}{K_1}\sqrt{\|(B_* - \hat B)H_X(B_* - \hat B)^\top \|_2\| (B_* + \hat B)H_X(B_* + \hat B)^\top \|_2},
\end{align}
where $H_X = \frac{1}{n}\sum_{i=1}^n X_iX_i^\top $, and the third inequality is because of the Cauchy-Schwarz inequality.

The first term in the right-hand side of \eqref{pfthmonelassoi1} $\|(B_* - \hat B)H_X(B_* - \hat B)^\top \|_2$ can be bounded by
\begin{align}\label{pfthmonelassoi1fir}
    \|(B_* - \hat B)H_X(B_* - \hat B)^\top \|_2 \leq & {\rm tr}((B_* - \hat B)H_X(B_* - \hat B)^\top )\nonumber\\
    \leq & d\max_{k}(\hat \beta_k - \beta_k^*)^\top H_X(\hat \beta_k - \beta_k^*)\nonumber\\
    \lesssim & ds^*\log p/n,
\end{align}
where tr$(A)$ is the trace of a matrix $A$, and the last inequality is by Lemma \ref{lemb3ning}. The second term in \eqref{pfthmonelassoi1} $(B_* + \hat B)H_X(B_* + \hat B)^\top $ can be bounded by
\begin{align}\label{pfthmonelassoi1sec}
    \|(B_* + \hat B)H_X(B_* + \hat B)^\top \|_2 = & \|(\hat B - B_* + 2B_*)H_X(\hat B - B_* + 2B_*)^\top \|_2\nonumber\\
    \leq & 2\|(\hat B - B_*)H_X(B_* - \hat B)^\top  + 4B_*H_X B_*^\top \|_2\nonumber\\
    \leq & 2\|(\hat B - B_*)H_X(B_* - \hat B)^\top \|_2 + 8\|B_*H_X B_*^\top \|_2\nonumber\\
    \lesssim & ds^*\log p/n + \|B_*H_X B_*^\top \|_2,
\end{align}
where the first inequality is by the Cauchy-Schwarz inequality, the second inequality is by the triangle inequality, and the third inequality is by \eqref{pfthmonelassoi1fir}.

Now consider bounding $\|B_*H_X B_*^\top \|_2$. By the triangle inequality, we have $\|B_*H_X B_*^\top \|_2\leq \|B_*\Sigma_{xx} B_*^\top \|_2 + \|B_*\Sigma_{xx} B_*^\top  - B_*H_X B_*^\top \|_2$. Therefore, we need to show that $\|B_*\Sigma_{xx} B_*^\top  - B_*H_X B_*^\top \|_2$ is small, which can be shown directly by Lemma \ref{lemconcenEqq}. To see this, note that $B_*X_i$ is still a sub-Gaussian random vector. Let $t = C_2 \sqrt{\log (d + p)/n}$ for some constant $C_2>0$ in Lemma \ref{lemconcenEqq}. By Lemma \ref{lemconcenEqq}, with probability at least $1 - d^2/p$, we have
\begin{align*}
    \|B_*\Sigma_{xx} B_*^\top  - B_*H_X B_*^\top \|_{\max} \lesssim \sqrt{\frac{\log (d+p)}{n}},
\end{align*}
which implies
\begin{align}\label{thmonelaasigaHcl}
    \|B_*\Sigma_{xx} B_*^\top  - B_*H_X B_*^\top \|_2 \leq d\|B_*\Sigma_{xx} B_*^\top  - B_*H_X B_*^\top \|_{\max} \lesssim d\sqrt{\frac{\log (d + p)}{n}}.
\end{align}
By Assumption \ref{onelassoassumcov}, Proposition \ref{propSigmaybound} and \eqref{covarlr}, we have
\begin{align*}
    \|B_*\Sigma_{xx} B_*^\top \|_2 = & \|\Sigma_{yy} - \Sigma_{\epsilon_y\epsilon_y}\|_2
    \leq \|\Sigma_{yy}\|_2 + \|\Sigma_{\epsilon_y\epsilon_y}\|_2\leq C_3,
\end{align*}
for some constant $C_3>0$. Therefore, with probability at least $1-d^2/p$, the right hand side of \eqref{pfthmonelassoi1sec} can be further bounded by
\begin{align}\label{pfthmonelassoi1sec2}
    \|(B_* + \hat B)H_X(B_* + \hat B)^\top \|_2 \lesssim & ds^*\log p/n + d\sqrt{\frac{\log (d+p)}{n}} + C_3.
\end{align}
Plugging \eqref{pfthmonelassoi1fir} and \eqref{pfthmonelassoi1sec2} into \eqref{pfthmonelassoi1}, we have 
\begin{align}\label{thmonelaI2bound}
    I_2\lesssim \sqrt{ds^*\log p/n}.
\end{align}
The first term $I_1$ in \eqref{thmonelatotal} can be bounded by
\begin{align}\label{thmonelaI1bound1}
    I_1 \leq & \bigg\|\bigg(\frac{1}{n}YY^\top \bigg)^{-1} - \Sigma_{yy}^{-1}\bigg\|_2\bigg\|\frac{1}{n}\hat BX(\hat BX)^\top \bigg\|_2.
\end{align}
By letting $t = C_4 \sqrt{\log (d + p)/n}$ for some constant $C_4>0$ in Lemma \ref{lemconcenEqq}, with probability at least $1 - d^2/p$, we have
\begin{align*}
    & \bigg\|\frac{1}{n}YY^\top  - \Sigma_{yy}\bigg\|_2 \leq d\bigg\|\frac{1}{n}YY^\top  - \Sigma_{yy}\bigg\|_{\max} \lesssim d\sqrt{\log (d+ p)/n}.
\end{align*}
For any unit vector $u$, by Lemma \ref{lemmaLinearSystems} and noting that Proposition \ref{propSigmaybound} implies $\kappa(\Sigma_{yy})\leq C_5$, we have 
\begin{align*}
    \bigg\|\bigg(\frac{1}{n}YY^\top \bigg)^{-1}u - \Sigma_{yy}^{-1}u\bigg\|_2 \lesssim d\sqrt{\log (d + p)/n},
\end{align*}
which implies
\begin{align}\label{thmonelaI1fir}
    \bigg\|\bigg(\frac{1}{n}YY^\top \bigg)^{-1} - \Sigma_{yy}^{-1}\bigg\|_2 \lesssim d\sqrt{\log (d + p)/n}.
\end{align}
The second term in the right-hand side of \eqref{thmonelaI1bound1} can be bounded by
\begin{align*}
   \bigg\|\frac{1}{n}\hat BX(\hat BX)^\top \bigg\|_2 = &\bigg\|\frac{1}{n}\hat BX(\hat BX)^\top  - \frac{1}{n} B_*X( B_*X)^\top  + \frac{1}{n} B_*X( B_*X)^\top \bigg\|_2 \nonumber\\
   \leq & \bigg\|\frac{1}{n}\hat BX(\hat BX)^\top  - \frac{1}{n} B_*X( B_*X)^\top \bigg\|_2  + \bigg\|\frac{1}{n} B_*X( B_*X)^\top \bigg\|_2,
\end{align*}
which can be bounded by a constant using the similar approach as in bounding $I_2$. Together with \eqref{thmonelaI1bound1} and \eqref{thmonelaI1fir}, we have 
\begin{align}\label{thmonelaI1bound}
    I_1\lesssim d\sqrt{\log (d + p)/n}.
\end{align}
By Proposition \ref{propSigmaybound} and \eqref{thmonelaasigaHcl}, it can be verified that the term $I_3$ can be bounded by 
\begin{align}\label{thmonelaI3bound}
    I_3 = & \bigg\|\Sigma_{yy}^{-1}\frac{1}{n}B_*X(B_*X)^\top  - \Sigma_{yy}^{-1}B_*\Sigma_{xx}B_*^\top \bigg\|_2\nonumber\\
    \leq &  \bigg\|\Sigma_{yy}^{-1}\bigg\|_2\bigg\|\frac{1}{n}B_*X(B_*X)^\top  - B_*\Sigma_{xx}B_*^\top \bigg\|_2\nonumber\\
    \lesssim & d\sqrt{\log (d + p)/n}.
\end{align}
Plugging \eqref{thmonelaI2bound}, \eqref{thmonelaI1bound} and \eqref{thmonelaI3bound} in \eqref{thmonelatotal}, we have 
\begin{align}\label{pfE}
    \|(YY^\top )^{-1}\hat BX(\hat BX)^\top  - \Sigma_{yy}^{-1}B_*\Sigma_{xx}B_*^\top \|_2 & \lesssim d\sqrt{\log (d + p)/n} + \sqrt{ds^*\log p/n}\nonumber\\
    & \lesssim d\sqrt{\log p/n} + \sqrt{ds^*\log p/n},
\end{align}
where the last inequality is because $d<p$.

To apply Lemma \ref{lemperteigenvec}, we need to show that $\|\delta\|_2$ in Lemma \ref{lemperteigenvec} is small. Let $E = (YY^\top )^{-1}\hat BX(\hat BX)^\top  - \Sigma_{yy}^{-1}B_*\Sigma_{xx}B_*^\top $. By \eqref{pfE}, we have
\begin{align*}
    \|\delta\|_2 = & \|Q_2^\top Eq_1\|_2 \leq \|Q_2^\top E\|_2\leq \|Q_2\|_2\|E\|_2 
    \lesssim \sqrt{d^2\log p/n} + \sqrt{ds^*\log p/n}\\
    \lesssim & \sqrt{d(d +s^*)\log p/n},
\end{align*}
where $Q$ is as in Assumption \ref{onelassoassumA}, and the last inequality follows the fact $\sqrt{w_1} + \sqrt{w_2}\leq \sqrt{2w_1 +2w_2}$ for $w_1,w_2>0$. By Assumption \ref{onelassoassumA}, we can apply Lemma \ref{lemperteigenvec}. This yields
\begin{align}\label{pfthmonelaaclose}
    \|a_1 - \hat a_1\|_2\lesssim \sqrt{d(d +s^*)\log p/n}.
\end{align}
For the second estimated canonical vector $\hat b_1$, we have
\begin{align}\label{pfthmonelab1}
    \|b_1 - \hat b_1\|_2 = & \bigg\|\frac{B_*^\top a_1}{\|B_*^\top a_1\|_2} - \frac{\hat B^\top \hat a_1}{\|\hat B^\top \hat a_1\|_2} \bigg\|_2 \lesssim \|B_*^\top a_1 - \hat B^\top \hat a\|_2\nonumber\\
    \leq & \|B_*^\top a_1 - B_*^\top \hat a_1\|_2 + \|B_*^\top \hat a_1- \hat B^\top \hat a_1\|_2\nonumber\\
    \leq & \|B_*^\top \|_2 \|a_1 - \hat a_1\|_2 + \|B_*^\top - \hat B^\top \|_2\|\hat a_1\|_2.
\end{align}
By Assumption \ref{oneassumB} and \eqref{pfthmonelaaclose}, 
\begin{align}\label{pfb1}
    \|B_*^\top \|_2 \|a_1 - \hat a_1\|_2 \lesssim \sqrt{d(d +s^*)\log p/n}.
\end{align}
Note that 
\begin{align}\label{pfthmonelab2}
    \|B_*^\top - \hat B^\top \|_2 \leq \|B_*^\top - \hat B^\top \|_F = \|B_*- \hat B\|_F  \lesssim \sqrt{\frac{ds^*\log p }{n}} ,
\end{align}
where the last inequality is because of Lemma \ref{lemb3ning}. By \eqref{pfthmonelaaclose}, $\|\hat a_1\|_2 \leq \|a_1 - \hat a_1\|_2 + \|a_1\|_2 \lesssim 1$. Therefore, combining \eqref{pfthmonelaaclose}, \eqref{pfthmonelab1}, \eqref{pfb1} and \eqref{pfthmonelab2} yields
\begin{align*}
    \|b - \hat b\|_2 \lesssim \sqrt{d(d +s^*)\log p/n} + \sqrt{\frac{ds^*\log p }{n}} \lesssim \sqrt{d(d +s^*)\log p/n}, 
\end{align*}
with probability at least $1-Cd^2/p$. The results of Theorem 2.1 follow the union bound inequality. Thus, we finish the proof.

\bibliography{refs}

\begin{thebibliography}{}

\bibitem[Aguet et~al., 2019]{aguet2019gtex}
Aguet, F., Barbeira, A.~N., Bonazzola, R., Brown, A., Castel, S.~E., Jo, B.,
  Kasela, S., Kim-Hellmuth, S., Liang, Y., Oliva, M., et~al. (2019).
\newblock The {GTE}x consortium atlas of genetic regulatory effects across
  human tissues.
\newblock {\em BioRxiv}, page 787903.

\bibitem[Barry et~al., 2018]{barry2018histopathological}
Barry, J.~D., Fagny, M., Paulson, J.~N., Aerts, H.~J., Platig, J., and
  Quackenbush, J. (2018).
\newblock Histopathological image {QTL} discovery of immune infiltration
  variants.
\newblock {\em iScience}, 5:80--89.

\bibitem[Chen et~al., 2012]{chen2012structure}
Chen, J., Bushman, F.~D., Lewis, J.~D., Wu, G.~D., and Li, H. (2012).
\newblock Structure-constrained sparse canonical correlation analysis with an
  application to microbiome data analysis.
\newblock {\em Biostatistics}, 14(2):244--258.

\bibitem[Chen et~al., 2013]{chen2013sparse}
Chen, M., Gao, C., Ren, Z., and Zhou, H.~H. (2013).
\newblock Sparse {CCA} via precision adjusted iterative thresholding.
\newblock {\em arXiv preprint arXiv:1311.6186}.

\bibitem[Cserh{\'a}ti et~al., 1998]{cserhati1998comparison}
Cserh{\'a}ti, T., K{\'o}sa, A., and Balogh, S. (1998).
\newblock Comparison of partial least-square method and canonical correlation
  analysis in a quantitative structure--retention relationship study.
\newblock {\em Journal of biochemical and biophysical methods},
  36(2-3):131--141.

\bibitem[Friedman et~al., 2010]{friedman2010regularization}
Friedman, J., Hastie, T., and Tibshirani, R. (2010).
\newblock Regularization paths for generalized linear models via coordinate
  descent.
\newblock {\em Journal of Statistical Software}, 33(1):1.

\bibitem[Gallins et~al., 2020]{paulzhou2020}
Gallins, P., Saghapour, E., and Zhou, Y.-H. (2020).
\newblock Exploring the limits of combined image/`omics analysis for non-cancer
  histological phenotypes.
\newblock {\em Frontiers in genetics}, doi:10.3389/fgene.2020.555886.

\bibitem[Gao et~al., 2017a]{gao2017sparse}
Gao, C., Ma, Z., Zhou, H.~H., et~al. (2017a).
\newblock Sparse {CCA}: {A}daptive estimation and computational barriers.
\newblock {\em The Annals of Statistics}, 45(5):2074--2101.

\bibitem[Gao et~al., 2017b]{gao2017discriminative}
Gao, L., Qi, L., Chen, E., and Guan, L. (2017b).
\newblock Discriminative multiple canonical correlation analysis for
  information fusion.
\newblock {\em IEEE Transactions on Image Processing}, 27(4):1951--1965.

\bibitem[Glahn, 1968]{glahn1968canonical}
Glahn, H.~R. (1968).
\newblock Canonical correlation and its relationship to discriminant analysis
  and multiple regression.
\newblock {\em Journal of the Atmospheric Sciences}, 25(1):23--31.

\bibitem[Gonz{\'a}lez et~al., 2008]{gonzalez2008cca}
Gonz{\'a}lez, I., D{\'e}jean, S., Martin, P.~G., and Baccini, A. (2008).
\newblock {CCA}: An {R} package to extend canonical correlation analysis.
\newblock {\em Journal of Statistical Software}, 23(12):1--14.

\bibitem[Grellmann et~al., 2015]{grellmann2015comparison}
Grellmann, C., Bitzer, S., Neumann, J., Westlye, L.~T., Andreassen, O.~A.,
  Villringer, A., and Horstmann, A. (2015).
\newblock Comparison of variants of canonical correlation analysis and partial
  least squares for combined analysis of mri and genetic data.
\newblock {\em Neuroimage}, 107:289--310.

\bibitem[Haghighi et~al., 2008]{haghighi2008learning}
Haghighi, A., Liang, P., Berg-Kirkpatrick, T., and Klein, D. (2008).
\newblock Learning bilingual lexicons from monolingual corpora.
\newblock In {\em Proceedings of ACL-08: Hlt}, pages 771--779.

\bibitem[Hardoon and Shawe-Taylor, 2011]{hardoon2011sparse}
Hardoon, D.~R. and Shawe-Taylor, J. (2011).
\newblock Sparse canonical correlation analysis.
\newblock {\em Machine Learning}, 83(3):331--353.

\bibitem[Horn and Johnson, 2012]{horn2012matrix}
Horn, R.~A. and Johnson, C.~R. (2012).
\newblock {\em Matrix Analysis}.
\newblock Cambridge University Press.

\bibitem[Hotelling, 1936]{hotelling1936relations}
Hotelling, H. (1936).
\newblock Relations between two sets of variates.
\newblock {\em Biometrika}.

\bibitem[Jordan et~al., 2013]{jordan2013statistics}
Jordan, M.~I. et~al. (2013).
\newblock On statistics, computation and scalability.
\newblock {\em Bernoulli}, 19(4):1378--1390.

\bibitem[L{\^e}~Cao et~al., 2009]{le2009sparse}
L{\^e}~Cao, K.-A., Martin, P.~G., Robert-Grani{\'e}, C., and Besse, P. (2009).
\newblock Sparse canonical methods for biological data integration:
  {A}pplication to a cross-platform study.
\newblock {\em BMC Bioinformatics}, 10(1):34.

\bibitem[Lee et~al., 2011a]{sccapackage}
Lee, W., Lee, D., Lee, Y., and Pawitan, Y. (2011a).
\newblock {\em scca: Sparse Canonical Covariance Analysis}.
\newblock R package version 1.1.1.

\bibitem[Lee et~al., 2011b]{lee2011sparse}
Lee, W., Lee, D., Lee, Y., and Pawitan, Y. (2011b).
\newblock Sparse canonical covariance analysis for high-throughput data.
\newblock {\em Statistical Applications in Genetics and Molecular Biology},
  10(1).

\bibitem[Li et~al., 2012]{li2012identifying}
Li, W., Zhang, S., Liu, C.-C., and Zhou, X.~J. (2012).
\newblock Identifying multi-layer gene regulatory modules from
  multi-dimensional genomic data.
\newblock {\em Bioinformatics}, 28(19):2458--2466.

\bibitem[Lutz and Eckert, 1994]{lutz1994relationship}
Lutz, J.~G. and Eckert, T.~L. (1994).
\newblock The relationship between canonical correlation analysis and
  multivariate multiple regression.
\newblock {\em Educational and Psychological Measurement}, 54(3):666--675.

\bibitem[Mai and Zhang, 2019]{mai2019iterative}
Mai, Q. and Zhang, X. (2019).
\newblock An iterative penalized least squares approach to sparse canonical
  correlation analysis.
\newblock {\em Biometrics}, 75(3):734--744.

\bibitem[Mardia et~al., 1979]{mardia1979multivariate}
Mardia, K.~V., Kent, J.~T., and Bibby, J.~M. (1979).
\newblock {\em Multivariate Analysis}.
\newblock Johns Hopkins University Press.

\bibitem[Moll et~al., 2008]{moll2008human}
Moll, R., Divo, M., and Langbein, L. (2008).
\newblock The human keratins: {B}iology and pathology.
\newblock {\em Histochemistry and Cell Biology}, 129(6):705.

\bibitem[Ning and Liu, 2017]{ning2014general}
Ning, Y. and Liu, H. (2017).
\newblock A general theory of hypothesis tests and confidence regions for
  sparse high dimensional models.
\newblock {\em The Annals of Statistics}, 45(1):158--195.

\bibitem[Park et~al., 2012]{park2012gplp}
Park, C., Huang, J.~Z., and Ding, Y. (2012).
\newblock Gplp: a local and parallel computation toolbox for {G}aussian process
  regression.
\newblock {\em The Journal of Machine Learning Research}, 13:775--779.

\bibitem[Parkhomenko et~al., 2007]{parkhomenko2007genome}
Parkhomenko, E., Tritchler, D., and Beyene, J. (2007).
\newblock Genome-wide sparse canonical correlation of gene expression with
  genotypes.
\newblock In {\em BMC Proceedings}, volume~1, page S119. Springer.

\bibitem[Parkhomenko et~al., 2009]{parkhomenko2009sparse}
Parkhomenko, E., Tritchler, D., and Beyene, J. (2009).
\newblock Sparse canonical correlation analysis with application to genomic
  data integration.
\newblock {\em Statistical Applications in Genetics and Molecular Biology},
  8(1):1--34.

\bibitem[Pau et~al., 2010]{pau2010ebimage}
Pau, G., Fuchs, F., Sklyar, O., Boutros, M., and Huber, W. (2010).
\newblock Ebimage—an {R} package for image processing with applications to
  cellular phenotypes.
\newblock {\em Bioinformatics}, 26(7):979--981.

\bibitem[Peng and Wu, 2014]{peng2014choice}
Peng, C.-Y. and Wu, C.~J. (2014).
\newblock On the choice of nugget in kriging modeling for deterministic
  computer experiments.
\newblock {\em Journal of Computational and Graphical Statistics},
  23(1):151--168.

\bibitem[Samarov et~al., 2011]{samarov2011local}
Samarov, D., Marron, J., Liu, Y., Grulke, C., and Tropsha, A. (2011).
\newblock Local kernel canonical correlation analysis with application to
  virtual drug screening.
\newblock {\em The Annals of Applied Statistics}, 5(3):2169.

\bibitem[Sargin et~al., 2007]{sargin2007audiovisual}
Sargin, M.~E., Yemez, Y., Erzin, E., and Tekalp, A.~M. (2007).
\newblock Audiovisual synchronization and fusion using canonical correlation
  analysis.
\newblock {\em IEEE Transactions on Multimedia}, 9(7):1396--1403.

\bibitem[Sass, 2012]{sass2012inborn}
Sass, J.~O. (2012).
\newblock Inborn errors of ketogenesis and ketone body utilization.
\newblock {\em Journal of Inherited Metabolic Disease}, 35(1):23--28.

\bibitem[Shu et~al., 2020a]{shu2020dGcca}
Shu, H., Qu, Z., and Zhu, H. (2020a).
\newblock D-gcca: Decomposition-based generalized canonical correlation
  analysis for multiple high-dimensional datasets.
\newblock {\em arXiv preprint arXiv:2001.02856}.

\bibitem[Shu et~al., 2020b]{shu2020d}
Shu, H., Wang, X., and Zhu, H. (2020b).
\newblock D-cca: A decomposition-based canonical correlation analysis for
  high-dimensional datasets.
\newblock {\em Journal of the American Statistical Association},
  115(529):292--306.

\bibitem[Song et~al., 2016]{song2016canonical}
Song, Y., Schreier, P.~J., Ram{\'\i}rez, D., and Hasija, T. (2016).
\newblock Canonical correlation analysis of high-dimensional data with very
  small sample support.
\newblock {\em Signal Processing}, 128:449--458.

\bibitem[Stein, 1999]{stein2012interpolation}
Stein, M.~L. (1999).
\newblock {\em Interpolation of Spatial Data: Some Theory for Kriging}.
\newblock Springer Science \& Business Media.

\bibitem[Stewart, 1973]{stewart1973error}
Stewart, G.~W. (1973).
\newblock Error and perturbation bounds for subspaces associated with certain
  eigenvalue problems.
\newblock {\em SIAM Review}, 15(4):727--764.

\bibitem[Suchard et~al., 2010]{suchard2010understanding}
Suchard, M.~A., Wang, Q., Chan, C., Frelinger, J., Cron, A., and West, M.
  (2010).
\newblock Understanding gpu programming for statistical computation: {S}tudies
  in massively parallel massive mixtures.
\newblock {\em Journal of computational and graphical statistics},
  19(2):419--438.

\bibitem[Sun et~al., 2009]{sun2009equivalence}
Sun, L., Ji, S., Yu, S., and Ye, J. (2009).
\newblock On the equivalence between canonical correlation analysis and
  orthonormalized partial least squares.
\newblock In {\em IJCAI}, volume~9, pages 1230--1235.

\bibitem[Tenenhaus and Guillemot, 2017]{rgccapackage}
Tenenhaus, A. and Guillemot, V. (2017).
\newblock {\em {RGCCA}: {R}egularized and {S}parse {G}eneralized {C}anonical
  {C}orrelation {A}nalysis for {M}ultiblock Data}.
\newblock R package version 2.1.2.

\bibitem[Tenenhaus and Tenenhaus, 2011]{tenenhaus2011regularized}
Tenenhaus, A. and Tenenhaus, M. (2011).
\newblock Regularized generalized canonical correlation analysis.
\newblock {\em Psychometrika}, 76(2):257.

\bibitem[Tibshirani, 1996]{tibshirani1996regression}
Tibshirani, R. (1996).
\newblock Regression shrinkage and selection via the lasso.
\newblock {\em Journal of the Royal Statistical Society: Series B
  (Methodological)}, 58(1):267--288.

\bibitem[Van~Loan and Golub, 1983]{van1983matrix}
Van~Loan, C.~F. and Golub, G.~H. (1983).
\newblock {\em Matrix Computations}.
\newblock Johns Hopkins University Press.

\bibitem[Vinokourov et~al., 2003]{vinokourov2003inferring}
Vinokourov, A., Cristianini, N., and Shawe-Taylor, J. (2003).
\newblock Inferring a semantic representation of text via cross-language
  correlation analysis.
\newblock In {\em Advances in neural information processing systems}, pages
  1497--1504.

\bibitem[Waaijenborg et~al., 2008]{waaijenborg2008quantifying}
Waaijenborg, S., de~Witt~Hamer, P. C.~V., and Zwinderman, A.~H. (2008).
\newblock Quantifying the association between gene expressions and dna-markers
  by penalized canonical correlation analysis.
\newblock {\em Statistical Applications in Genetics and Molecular Biology},
  7(1).

\bibitem[Wang et~al., 2015]{wang2015inferring}
Wang, Y.~R., Jiang, K., Feldman, L.~J., Bickel, P.~J., Huang, H., et~al.
  (2015).
\newblock Inferring gene--gene interactions and functional modules using sparse
  canonical correlation analysis.
\newblock {\em The Annals of Applied Statistics}, 9(1):300--323.

\bibitem[Witten and Tibshirani, 2020]{pmapackage}
Witten, D. and Tibshirani, R. (2020).
\newblock {\em PMA: Penalized Multivariate Analysis}.
\newblock R package version 1.2.1.

\bibitem[Witten et~al., 2009]{witten2009penalized}
Witten, D.~M., Tibshirani, R., and Hastie, T. (2009).
\newblock A penalized matrix decomposition, with applications to sparse
  principal components and canonical correlation analysis.
\newblock {\em Biostatistics}, 10(3):515--534.

\bibitem[Witten and Tibshirani, 2009]{witten2009extensions}
Witten, D.~M. and Tibshirani, R.~J. (2009).
\newblock Extensions of sparse canonical correlation analysis with applications
  to genomic data.
\newblock {\em Statistical Applications in Genetics and Molecular Biology},
  8(1):1--27.

\bibitem[Yamamoto et~al., 2008]{yamamoto2008canonical}
Yamamoto, H., Yamaji, H., Fukusaki, E., Ohno, H., and Fukuda, H. (2008).
\newblock Canonical correlation analysis for multivariate regression and its
  application to metabolic fingerprinting.
\newblock {\em Biochemical Engineering Journal}, 40(2):199--204.

\bibitem[Yazici et~al., 2010]{yazici2010application}
Yazici, A.~C., {\"O}{\u{g}}{\"u}{\c{s}}, E., Ankarali, H., and G{\"u}rb{\"u}z,
  F. (2010).
\newblock An application of nonlinear canonical correlation analysis on medical
  data.
\newblock {\em Turkish Journal of Medical Sciences}, 40(3):503--510.

\end{thebibliography}

\end{document}